\documentclass{IEEEtran}
\usepackage{cite}
\usepackage{amsmath,amssymb,amsfonts}
\usepackage{algorithmic}
\usepackage{graphicx}
\usepackage{textcomp}

\usepackage{stmaryrd}
\usepackage{stackrel}
\usepackage{amssymb}
\usepackage{bbm}
\usepackage{amsmath}
\usepackage{amsfonts}   
\usepackage[ruled,vlined]{algorithm2e}
\usepackage{graphicx}
\usepackage{amsthm}
\usepackage{mathtools}

\newcommand{\bbP}{\mathbb{P}}

\newcommand{\mh}{\hat{\mu}}

\newtheorem{proposition}{Proposition}
\newtheorem{lemma}{Lemma}
\newtheorem{claim}{Claim}
\newtheorem{theorem}{Theorem}

\theoremstyle{definition}
\newtheorem{remark}{Remark}
\newtheorem{definition}{Definition}

\DeclarePairedDelimiter\ceil{\lceil}{\rceil}

\def\BibTeX{{\rm B\kern-.05em{\sc i\kern-.025em b}\kern-.08em
    T\kern-.1667em\lower.7ex\hbox{E}\kern-.125emX}}
\begin{document}
\title{Optimal Learning for Dynamic Coding in Deadline-Constrained Multi-Channel Networks}
\author{Semih Cayci, \IEEEmembership{Student Member, IEEE} and Atilla Eryilmaz, \IEEEmembership{Senior Member, IEEE}
\thanks{This work is funded by the NSF grants: CCSS-EARS-1444026, CNS-NeTS-1514127, CMMI-SMOR-1562065 and CNS-WiFiUS-1456806, CNS-ICN-WEN-1719371; the DTRA grant HDTRA1-15-1-0003; and the QNRF Grant NPRP 7-923-2-344.}
\thanks{Semih Cayci and Atilla Eryilmaz are with the Ohio State University, Columbus, OH 43210 USA (e-mail: \{cayci.1, eryilmaz.2\}@osu.edu). }}

\maketitle

\begin{abstract}
We study the problem of serving randomly arriving and delay-sensitive traffic over a multi-channel communication system with time-varying channel states and unknown statistics. This problem deviates from the classical exploration-exploitation setting in that the design and analysis must accommodate the dynamics of packet availability and urgency as well as the cost of each channel use at the time of decision. To that end, we have developed and investigated an index-based policy UCB-Deadline, which performs dynamic channel allocation decisions that incorporate these traffic requirements and costs. Under symmetric channel conditions, we have proved that the UCB-Deadline policy can achieve bounded regret in the likely case where the cost of using a channel is not too high to prevent all transmissions, and logarithmic regret otherwise. In this case, we show that UCB-Deadline is order-optimal. We also perform numerical investigations to validate the theoretical findings, and also compare the performance of the UCB-Deadline to another learning algorithm that we propose based on Thompson Sampling.
\end{abstract}


\section{Introduction}
\label{sec:introduction}
\IEEEPARstart{W}{ith} the advances in wireless communications, next generation communication networks are expected to serve real-time applications that require end-to-end deadline constraints and a large amount of throughput over fading channels. Especially real-time multimedia applications such as voice and video streaming possess stringent deadline constraints that require particular emphasis. The ultra-wideband communication channels that are designed to meet these requirements, such as millimeter-wave (mmW) channels, have highly intermittent dynamics, which makes existing channel probing and estimation techniques inapplicable. Therefore, it is crucial to develop new communication schemes that can handle applications with deadline constraints and large throughput demands in the absence of channel statistics and channel state information.

In wireless communication schemes such as IEEE 802.11 and 5G millimeter-wave (mmW) cellular systems, availability of multiple orthogonal channels enables a user to simultaneously utilize multiple channels to increase the quality of communication in various aspects \cite{Vaidya, Rappaport}. In \cite{Vaidya}, it is shown that multi-channel operation provides significant increase in network capacity, which can be exploited to meet the increasing demand for throughput. In mmW cellular communications, multi-channel scenario is expected to overcome the intermittence problem of mmW channels due to blockage, which particularly hinders applications with quality of service (QoS) requirements \cite{Rappaport, Cayci}. As it is possible to equip a single node with multiple radio interfaces due to the reduced hardware costs, multi-channel communication scheme offers a feasible solution to serve applications with deadline constraints and large throughput demand \cite{Vaidya, Cost}. On the other hand, operational costs, such as power consumption, impose a critical constraint in the number of active interfaces. Thus, it is important to activate a plausible number of channels dynamically depending on queue-length and deadline constraints so as to increase throughput while keeping the operational costs at acceptable levels.

In conventional communication systems, there are efficient channel estimation techniques that provide channel state information (CSI) for rate and power allocation policies \cite{Tse}. However, these methods are inapplicable in millimeter-wave communication systems as the channels are highly intermittent and fast-varying \cite{Rappaport, Rangan}. This necessitates the development of online learning algorithms that rely on channel feedback in the absence of channel state information and channel statistics.

In this paper, we investigate the problem of dynamic channel allocation for a single user in a multi-channel network with deadline constraints and service costs in the absence of channel statistics and CSI. Our main contribution is an online learning algorithm that converges to the optimal solutions with small regret by using only the channel feedback. In traditional communication systems, efficient rate and power allocation schemes that base the decisions on CSI and queue-lengths exist \cite{NeelyModiano, EryilmazPerkins, NeelyGeorgiadis, Shroff, Gangammanavar}. However, these methods are built on the key assumption that CSI is available at the time of decision, therefore they are not applicable in the emerging communication scenarios where CSI and channel statistics are unknown. There is an interesting body of work which considers the online learning problem for rate allocation based on success/fail feedback \cite{Agarwal1, Agarwal2}. These works do not apply to our context since they do not provide short-term performance guarantees, such as regret.

There is a large body of work in the design and analysis of online learning algorithms that optimize short-term performance in the context of multi-armed bandits (MAB) \cite{Bubeck2012, Tekin}. Our work deviates from the context of classical stochastic bandits as the revenue of the activated arms are coupled, the controller has the incentive to activate no channels due to the cost, and there is a strong dependence on the queue-length. In \cite{Shakkotai}, learning problem is investigated with a regret definition based on queueing-delay. This work does not apply to our setting as it does not consider deadline-constrained traffic. A preliminary version of this work was presented in \cite{Cayci_Learning}, where the deadline constraint is fixed at one time-slot and the throughput is defined as the number of successfully transmitted packets. In this paper, we generalize the results to any deadline under an erasure coding scheme.

\section{Dynamic Channel Encoding-Decoding and Learning Problem}
\label{sec:sys_mod}
We consider a discrete-time multi-timescaled system consisting of frames and time-slots. Time-slot is the smallest time unit in this framework in which channel variations occur, and each frame consists of $T\geq 1$ time-slots.

\textbf{Channel Process:} We study a multi-channel system in which the packets can be transmitted by $K$ (possibly infinite) independent fading channels. In frame $n$, the rate $C_k^n(t)$ of channel $k$ evolves in each time-slot according to an iid Bernoulli process with mean $\mu$, i.e., $C_k^n(t)\overset{iid}{\sim}Ber(\mu^*)$ for $k=1,2,\ldots,K$. This Bernoulli channel model reflects the sharp difference between line-of-sight (LOS) and non-line-of-sight channel (NLOS) states in millimeter-wave communications \cite{Rappaport, Rangan}. $C_k^n(t)$ is revealed via ACK or NACK signals after the transmission only if channel $k$ is activated at time $t$. In the learning problem, $\mu^*$ is not known a priori, and learned over time by using feedback. We assume that the belief about $\mu^*$ is updated only at the beginning of each frame.

\textbf{Arrival Process:} The packets arrive into the system only at the beginning of each frame according to an arrival process $A(n)$ which is independent and identically distributed (iid) over a finite set $\mathcal{A}=\{0,1,\ldots,A_{max}\}$ with probability distribution $\bbP(A(n)=a)=\alpha_a$ for $a\in\mathcal{A}$ at frame $n$. The packets have a lifetime of one frame, i.e. $T$ time-slots, and will be lost if they are not served within that interval.

The overall problem consists of two subproblems with different time-scales: a fast timescale problem of rate allocation and a slow timescale problem of learning. The fast timescale system is concerned about encoder-decoder couple selection at each time-slot within a frame given a belief on the channel statistics. The slow timescale system focuses on the learning part, and the goal is to update the belief on the channel statistics using the feedback so as to maximize performance. The overall system model is illustrated in Figure \ref{fig:sys_scheme}.
\begin{figure}[htb]
\centering
\includegraphics[scale=0.25]{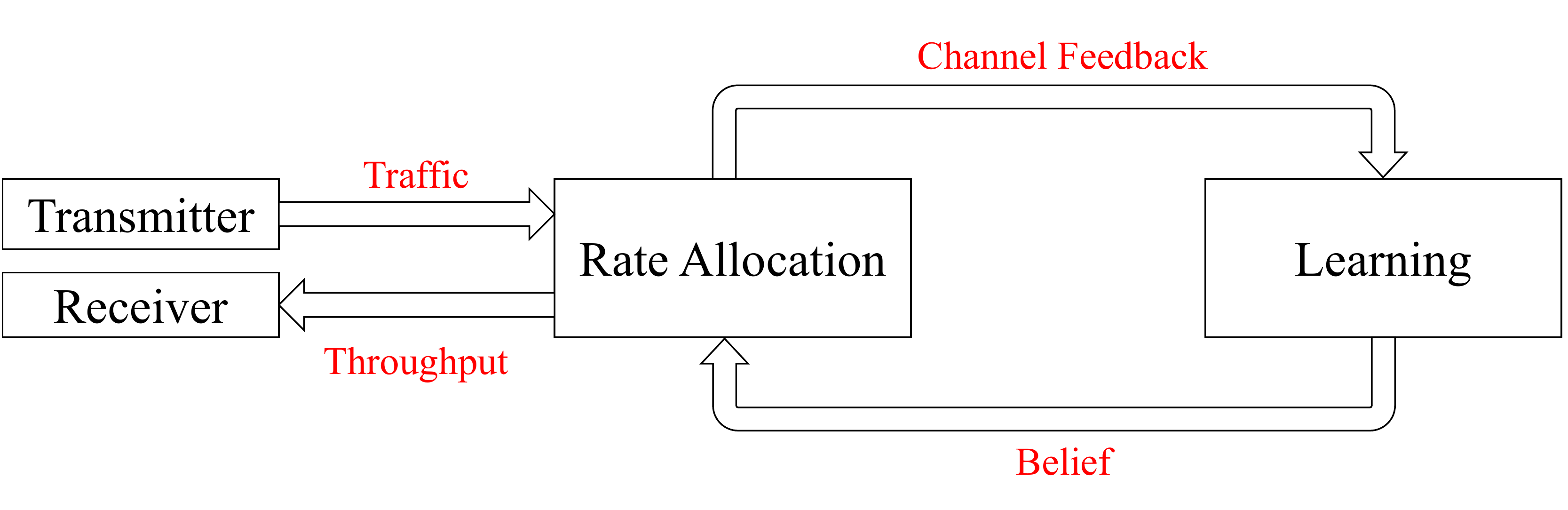}
\caption{The overall system model. The rate allocation policy provides the ACK/NACK feedback to the learning policy, and receives the belief about channel statistics in return.}
\label{fig:sys_scheme}
\end{figure}

In the next subsection, we investigate the dynamic channel encoding-decoding problem that takes place in the fast timescale.

\subsection{Dynamic Channel Encoding-Decoding Problem}
In rate allocation, we focus on a single frame $n$ given a channel rate estimate $\mu$. At each stage $s\in\{1,2,\ldots,T\}$, a centralized controller selects $x_s$ packets from the queue, and transmits them over $m_s$ channels by choosing an encoder-decoder pair of an $(m_s, x_s)$ code. The system is illustrated in Figure \ref{fig:system_model}.
\begin{figure}[htb]
\centering
\includegraphics[scale=0.34]{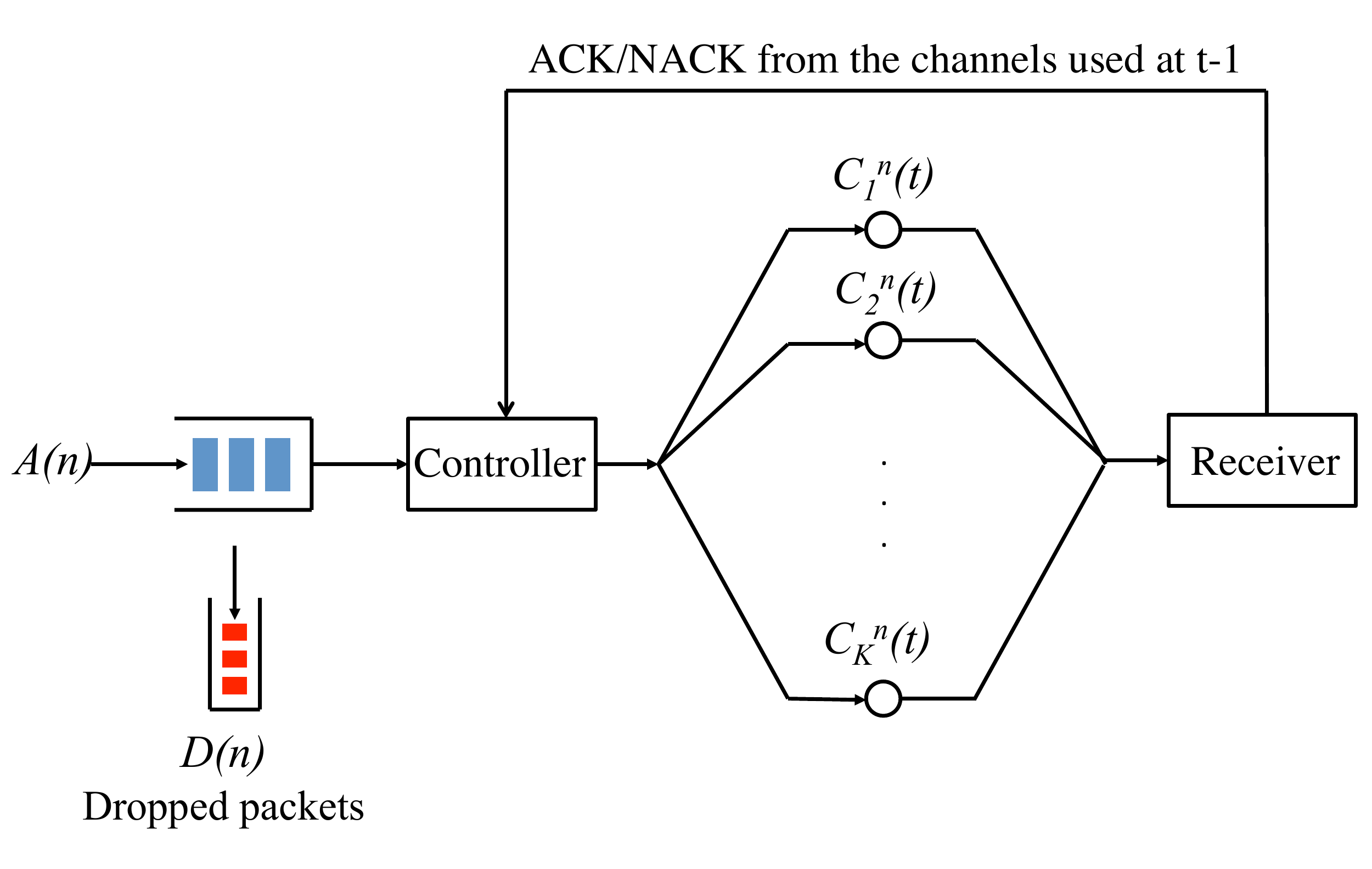}
\caption{The multi-channel network with symmetric Bernoulli channels. At the end of each transmission, CSI of each activated channel is revealed to the controller via ACK/NACK signals.}
\label{fig:system_model}
\end{figure}

Each channel use incurs a constant cost of $d\in[0,1]$, which measures the operational costs associated with each channel use, such as power. It is assumed that $d$ is known by the controller. If $m$ channels are activated and there are $x$ packets scheduled for transmission at time-slot $t$, the throughput is denoted by $\tau_{(m, x)}(t)$, and it is a function of the states of the activated channels. With these definitions, the problem corresponds to choosing an encoder-decoder couple for an $(m, x)$ code to maximize the total revenue in a frame. Let $V_{(m, x)}(t) = \tau_{(m, x)}(t)-d\cdot m$ be the revenue at time-slot $t$. The process is a controlled Markov chain with controls $(m_s,x_s)$ at stage $s$, and state $X_s$, which denotes the number of remaining packets at the beginning of stage $s$. The state transition occurs in the following sense:
\begin{equation*}
X_{s-1} = X_s - \tau_{(m_s,x_s)}(T-s+1).
\end{equation*}
 Assuming $A(n) = a$, the total revenue under a policy $\pi_a$ is as follows:
\begin{equation}
\label{eqn:value_function}
		J(a, \pi_a) = \sum_{t=1}^T V_{(m_t,x_t)}(t) + J_0(a-\sum_{t=1}^T \tau_{(m_t,x_t)}(t) ),
\end{equation}


\noindent where $\pi_a=(\pi_a^1,\pi_a^2,\ldots,\pi_a^T)$ is a rate allocation policy and $J_0:\mathcal{A}\rightarrow\mathbb{R}^-$ is a negative and monotonically decreasing penalty function for untransmitted packets. In this work, we will consider a linear penalty function $J_0(X) = -\lambda X$ for some $\lambda \geq 0$. This corresponds to penalizing each untransmitted packet by $\lambda$. 

Let $\mathcal{U}_n(t)$ denote the set of channels activated at time-slot $t$ in frame $n$ and $h_t = [(m_s, x_s, \mathcal{U}_{n}(s))]_{s=1}^{t-1}$ be the history up to time $t$. An admissible policy $\pi$ chooses an $(m_t,x_t)$ code based on the history $h_t$ at time-slot $t$, i.e., $\pi_a^t:h_t\mapsto (m_t,x_t)$. Here, in the absence of the knowledge of $\mu^*$, the rate allocation policy makes a decision based on a given belief $\mu$ about $\mu^*$, i.e., $\pi_a = \pi_a(\mu)$. Our goal is to find the admissible policy that achieves the maximum total revenue based on the current belief $\mu$ given $A(n) = a$, denoted by $\pi_a^*(\mu)$, which is the solution to the following optimization problem:
\begin{equation}
\label{eqn:opt_rate_alloc}
\underset{\pi_a}{\max}~\mathbb{E}_{\mu}^{\pi_a}\Big[ J(a, \pi_a) \Big] ,
\end{equation}

\noindent where the expectation is taken with respect to $\mu$. In next section, we will provide an optimal solution to the rate allocation problem by using dynamic programming.

Throughout this paper, we will consider a  specific communication scenario that uses a near-optimal erasure coding scheme with the following throughput function:
\begin{equation}
\tau_{(m_t,x_t)}(t) = x_t\cdot \mathbb{I}_{ \{ S_{m_t}(t) \geq x_t \} }.
\end{equation}
\noindent where $S_m(t) = C_1^n(t)+C_2^n(t)+\ldots+C_m^n(t)$ is the number of connected channels when $m$ channels are activated. This communication scenario applies to a broad variety of wireless and optical communication scenarios as well as storage applications \cite{Costello}.

%


%

\subsection{Learning Problem}

In rate allocation, we assumed a given belief $\mu$ about the channel parameter $\mu^*$. Remember that we do not have an a priori knowledge about $\mu^*$ at the beginning, and we learn the channel statistics by using the channel feedback. In order to maximize the revenue over time, it is required to have a reliable estimate on $\mu^*$ for the rate allocation policy, which necessitates a learning policy that leads to a fast convergence.

Let $\bar{\mu}_n$ be an estimator of $\mu^*$. We say that $\bar{\mu}_n$ is admissible if it is based on the knowledge of activated channel realizations until and excluding frame $n$, and arrivals until and including $n$:
\begin{multline*}
\mathbb{I}_{\{\bar{\mu}_n = \mu \}}\in\sigma(\{C_i^k(s): i \in \mathcal{U}_k(s), s\leq T, k < n\}, \\ \{A(k):k\leq n\}),
\label{eqn:admissibility}
\end{multline*}
\noindent for all $\mu\in  [0,1]$, where $\mathbb{I}$ is the indicator function, $\mathcal{U}_n(t)$ denotes the set of channels activated at time-slot $t$ in frame $n$ and $\sigma(\{X_i\}_{i=1}^J)$ denotes the $\sigma$-field generated by a collection of random variables $X_i,~i=1,2,\ldots, J$.

Therefore, we can define an admissible joint learning and rate allocation policy as follows:
\begin{equation}
\pi_{A(n)}(\bar{\mu}_n)=\sum_{a\in\mathcal{A}}\mathbb{I}_{\{ A(n)=a \}}\pi_a(\bar{\mu}_n),
\end{equation}

\noindent where $\bar{\mu}_n$ is an admissible estimator and $\pi_a$ is an admissible rate allocation policy.

%
%
%

\subsection{Regret of an Admissible Policy}
Recall that if a genie reveals the mean $\mu^*$ to the controller, the optimal rate allocation policy that maximizes the revenue given $A(n) = a$ is $\pi^*_a(\mu^*)$. As the a priori knowledge of $\mu$ is absent, an algorithm has to learn the mean, and maximize the revenue simultaneously. Pseudo-regret, which will be simply referred to as regret throughout this paper, is a common measure to evaluate the performance of learning algorithms \cite{Bubeck2012, CesaBianchi, Tekin}. The regret under an admissible policy $\pi_{A(n)}(\bar{\mu}_n)$ for a horizon $N$ is defined as follows:
\begin{equation*}
\bar{R}_N =\mathbb{E} \sum_{n=1}^N \sum_{a\in\mathcal{A}} \mathbb{I}_{\{A(n)=a\}} ( J(a, \pi_a^*(\mu^*))-J(a, \pi_a(\bar{\mu}_n)) ).
\end{equation*}

\noindent In words, regret is defined as the cumulative difference between the maximum expected revenue given the mean $\mu^*$ and the expected revenue under policy $\pi$ in $N$ frames.

By using the admissibility of the policy $\pi$, the regret can be found as follows:
\begin{equation}
 \label{eqn:regret}
\bar{R}_N = \sum_{a\in\mathcal{A}} \alpha_a \cdot \mathbb{E}_{\mu^*} \Big[\sum_{n=1}^N  \Big ( J(a, \pi_a^*(\mu^*))-J(a, \pi_a(\bar{\mu}_n))\Big ) \Big].
\end{equation}

The objective in this paper is to design policies that provide low regret. In the following section, we first investigate the optimal rate allocation policy along with its characteristics, and then propose learning algorithms that provably achieve low regret.

\section{Properties of Optimal Rate Allocation Policy}
\label{sec:prop}
In this section, we will investigate the optimal rate allocation policy and provide some important characteristics under a specific channel assumption. We first describe a method based on dynamic programming to find the optimal rate allocation policy.

\subsection{Rate Allocation as a Dynamic Programming Problem}

Assume that the the controller has the belief $\mu$ for the channel mean. Let $P_{m,x}(\mu) = \mathbb{P}(S_m^n(T-s+1) \geq x)$ be the probability of successful transmission where $S_m^n(t)$ is the total number of connected channels to transmit $x$ packets over $m$ channels. Then, the optimization problem in frame $n$ in (\ref{eqn:opt_rate_alloc}) can be recast as Bellman-Ford recursions as follows:

\begin{multline}
\label{eqn:dp}
J_s(X_s) = \underset{x\leq X_s, m}{\max}\{-dm + P_{m,x}(\mu)\cdot x\\+P_{m, x}(\mu)\cdot J_{s-1}(X-x) + (1-P_{m,x}(\mu))\cdot J_{s-1}(X)\},
\end{multline}
\noindent where the terminal cost $J_0$ is the penalty function for untransmitted packets as before, the state variable $X_s$ denotes the number of remaining packets at the beginning of stage $s$, $X_T=A(n)$.

Therefore, the optimal rate allocation policy can be stated in the following way \cite{Whittle}:
\begin{multline}
(m_s, x_s) = \underset{(m,x):x\leq X_s}{\arg\max} \{-dm \\+ P_{m,x}(\mu)\Big(x+J_{s-1}(X_s-x)-J_{s-1}(X_s)\Big) \},
\end{multline}
for $s=1, 2,\ldots, T$ with the initial condition $J_0(X) = -\lambda X$ for all $X\in \mathcal{A}$.

\begin{remark}
Note that the value function in (\ref{eqn:dp}) is not the actual value function as it assumes that the channel mean is the belief $\mu$, which is generally different from the actual mean $\mu^*$. The average performance of a rate allocation policy $(m_t,x_t)$ at frame $n$ given a belief $\bar{\mu}_n$ becomes as follows:
\begin{equation}
\mathbb{E}[J(A(n), \pi_{A(n)}(\bar{\mu}_n)] = \sum_{a\in\mathcal{A}} \alpha_a \mathbb{E}_{\mu^*}[J(a, \pi_{a}(\bar{\mu}_n))],
\end{equation} 
\noindent where $J$ is defined in (\ref{eqn:value_function}).
\end{remark}

\subsection{Characteristics of the Optimal Policy}
We now investigate some important characteristics of the optimal rate allocation policy that will be crucial in the performance analysis of the learning policy. Throughout this discussion, we assume that $\mu^*$ is provided by a genie.

\begin{proposition}[Critical point]
\label{prop:cp}
Under optimal rate allocation policy, there exists a critical $\zeta \in [\frac{d}{1+\lambda}, \frac{2d}{1+\lambda}]$ such that $\sum_{t=1}^T m_t > 0$ if and only if $\mu^* \geq \zeta$.
\end{proposition}
\noindent Proof of Proposition \ref{prop:cp} is given in the Appendix.

In regret analysis, it is crucial to find upper bounds on instantaneous regret, i.e., mean-independent lower bounds on the expected revenue. The following proposition provides such a finite lower bound on the expected revenue.

\begin{proposition}
\label{prop:bdd}
For any $a\in\mathcal{A}$, there exists $B_a = B_a(\mu^*)>0$ such that $\mathbb{E}_{\mu^*}[J(a, \pi_a(\mu))] \geq -B_a$ for all $\mu \in (0, 1)$.
\end{proposition}
\begin{proof}
At any time-slot, the throughput is upper bounded by $a$, therefore at most $a/d$ channels are used. In that case, a very crude lower bound would be when throughput is $0$ although all channels are used, which implies $\mathbb{E}_{\mu^*}[J(a, \pi_a(\mu))] \geq -Ta/d-\lambda a$.
\end{proof}

\begin{remark}
By the same argument as Proposition \ref{prop:bdd}, it is straightforward to show that $\sum_{t=1}^T m_t$ is upper bounded if $\mathcal{A}$ is bounded: $\sum_{t=1}^Tm_t \leq M_{max} = TA_{max}/d$.
\end{remark}

These characteristics will be fundamental in the analysis of learning algorithms that will be presented in the next section.

\subsection{Case Study: Delay-Tolerant and Delay-Intolerant Systems}
We now investigate two extreme cases that will provide important insights about the characteristics of the optimal rate allocation policy: delay-tolerant and delay-intolerant systems.

\textbf{Case I: Delay-Tolerant System}: In the first extreme case, we assume that there is at most one packet in the queue, i.e., $A_{max} =1$, and multiple time-slots for the successful transmission of that packet, i.e., $T\geq 1$. This corresponds to a delay-tolerant system where each packet has multiple time-slots for transmission.

The optimal rate allocation policy can be found by solving the Bellman-Ford recursions given in (\ref{eqn:dp}). Starting with $J_0(1) = -\lambda$, the optimal rate allocation policy can be found as follows:
\begin{equation}
\label{eqn:dt_opt}
m_t = \left\{
        \begin{array}{ll}
            0 ,&\mbox{if } \mu^* < \frac{d}{1-J_{t-1}(1)} \\
            k ,&\mbox{if }  (1-\mu^*)^{k}\mu^* < \frac{d}{1-J_{t-1}(1)} < (1-\mu^*)^{k-1}\mu^*
        \end{array}
    \right.
\end{equation}
\noindent for $k=1,2,\ldots$, where $x_t = 0$ if $m_t = 0$ and $x_t = 1$ otherwise.

\textbf{Example:} In the following, we investigate how the optimal rate allocation policy varies with different values of $\mu^*\in[0,1]$ in a specific setting. Figure \ref{fig:dl_tol} demonstrates $\{m_t:t=1,2,\ldots, T\}$ where $T = 4$, $d=0.25$ and $\lambda=1$.
\begin{figure}[htb]
\centering
\includegraphics[scale=0.4]{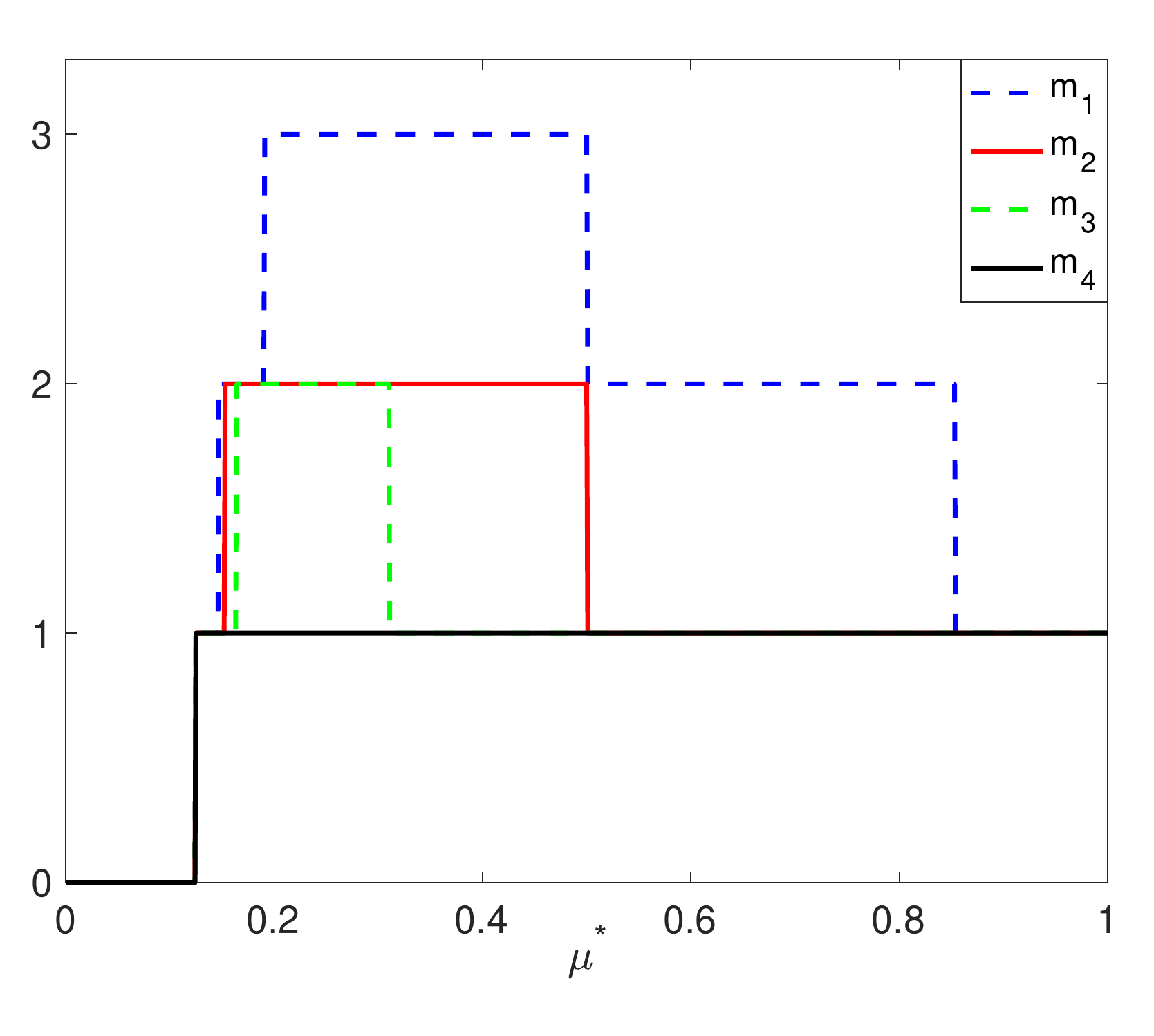}
\caption{Optimal policy as a function of $\mu$ for the case $A_{max}=1$, $d=0.25$, $T=4$ and $\lambda=1$. The channel use becomes more aggressive as the deadline approaches.}
\label{fig:dl_tol}
\end{figure}

\noindent From Figure \ref{fig:dl_tol}, we first observe that the number of activated channels increases as the deadline approaches for any $\mu^*$. Secondly, at any time-slot, the number of activated channels increases up to a certain $\mu^*$, and then monotonically decreases with the increasing $\mu^*$. This stems from the fact that an additional channel is costly when the reliability of current channels is high enough.

There are three important observations that stem from (\ref{eqn:dp}) and (\ref{eqn:dt_opt}).

\begin{proposition}
\begin{enumerate}
\item The value function increases over time: $J_t(1) \geq J_{t-1}(1)$ for $t=1,2,\ldots,T$.
\item The number of channel uses increases as the deadline approaches: $m_t \leq m_{t-1}$ for $t=1,2,\ldots,T$.
\item The critical point is $\zeta = \frac{d}{1+\lambda}$, which is the lower limit given in Prop. \ref{prop:cp}.
\end{enumerate}
\end{proposition}
\begin{proof}
\noindent The first claim follows since $m_t = 0,~\forall t$ is among the possible decisions. The second claim follows from (\ref{eqn:dt_opt}), which says that $m_t$ is monotonically decreasing with $J_t$, which is monotonically increasing with $t$.
\end{proof}

\textbf{Case II: Delay-Intolerant System}: In the delay-intolerant case, there are bursty arrivals with $A_{max}>1$ that required to be transmitted in only one time-slot in each frame, i.e., $T=1$.

In order to analyze the optimal rate allocation policy in this case, we will assume that channel use is feasible, i.e., $\mu^* > \zeta$ where $\zeta$ is the critical point, and devise a continuous approximation via central limit theorem. We first define the continuous approximation of the value function.

\begin{definition}
Let $m,x\in\mathbb{R}$ and $\Phi(z) = \mathbb{P}(Z<z)$ where $Z\sim \mathcal{N}(0, 1)$. For any $a\in\mathcal{A}$ and $\lambda \geq 0$, the continuous approximation of the value function $J_1(a)$ is defined as follows:
\begin{equation}
\tilde{J}_1(a) = -\lambda a+\underset{x\leq a, m}{\max} ~\nu(m,x),
\end{equation}
\noindent where $\nu(m,x)=-dm+x\Phi\Big(\frac{x-m\mu}{\sqrt{m\mu(1-\mu)}}\Big)(1+\lambda)$.
\end{definition}
 
 \noindent Note that $\tilde{J}_1(X)$ is a good approximation of $J_1(X)$ if $x_1$ is large and $\mu^* \gg\zeta$, which together imply that $m_1$ is large. Also, it is straightforward to show that $\nu(m,x)$ is unimodal in $x$ for a fixed $m$ and also unimodal in $m$ for a fixed $x$.
 
 In the following, we show that if $\mu^* \gg\zeta$, then $x_1=a$, i.e., it is optimal to transmit all the packets in the queue, and provide important characteristics of optimal coding rate. For a given word-length $x$, the following lemma provides a way to find the optimal block-length, which will be very useful in finding the optimal word-length for transmission and characterizing the code rate.
 
\begin{lemma}
\label{lemma:opt}
For a given $x$, let $M^*(x) = \arg\underset{m}{\max} ~\nu(m,x)$ be the optimal set. If $\mu^* \geq \zeta$, $M^*(x)$ consists of exactly one element and the unique maximizer $m\in M^*(x)$ satisfies the following equation:
\begin{equation}
\label{eqn:opt_mx}
	(1+\lambda)\frac{x}{\sqrt{m\mu^*(1-\mu^*)}}\varphi\Big(\frac{x-m\mu^*}{\sqrt{m\mu^*(1-\mu^*)}}\Big) = \frac{2d}{r+\mu^*},
\end{equation}
\noindent where $r = \frac{x}{m}$ is the code rate and $\varphi(z) = \frac{d}{dz}\Phi(z)$ is the density function of a standard Gaussian random variable.
\end{lemma}

\begin{proof}
	For a fixed $x$, it is straightforward to show that $\nu(m,x)$ is a smooth and unimodal function of $m$. Therefore, there is a unique maximizer which can be found by solving $\frac{\partial}{\partial m}\nu(m,x) = 0$.
\end{proof}

In the next proposition, we show that the optimal code rate is strictly below $\mu^*$.

\begin{proposition}[Characterization of the Code Rate]
\label{prop:rate}
Assume that $\mu^*\geq \frac{2d}{1+\lambda}$. For any $x\geq \pi/2$, let the optimal code rate be denoted as $r(x) = x/m$ for $m\in M^*(x)$. Then, there exists $\delta = \delta(x)>0$ such that $r(x) = \mu^*-\delta$.
\end{proposition}

\begin{proof}
Fix $x\geq 1$. Note that the function $\nu(m,x)$ is unimodal for $m\geq 0$ and the assumption $\mu^* \geq \frac{2d}{1+\lambda}$ implies that $\nu(\frac{x}{\mu^*},x) > 0$ by Proposition \ref{prop:cp}. Then,
\begin{eqnarray*}
\frac{\partial}{\partial m}\nu(\frac{x}{\mu^*}, x) &=& (1+\lambda)\mu^*\sqrt{\frac{x}{2\pi(1-\mu^*)}} - d, \\
&\geq& d\cdot\Big(\sqrt{\frac{2x}{\pi(1-\mu^*)}}-1\Big), \\
&>& 0,
\end{eqnarray*}
\noindent where the second line follows from the assumption $\mu^* \geq \frac{2d}{1+\lambda}$ and the last line follows since $x\geq \pi/2$. Together with this, the unimodality implies that the optimal point $m$ is strictly above $\frac{x}{\mu^*}$. Hence, $r(x) < \mu^*$, i.e., the optimal code rate is strictly below $\mu^*$.
\end{proof}

\begin{remark}
The result of Proposition \ref{prop:rate} is very intuitive: for a fixed word-length $x$, a slight increase in the block-length $m$ beyond $x/\mu^*$ leads to an exponential increase in the probability of success at the expense of only a linear increase in the cost, which makes a slightly lower code rate than $\mu^*$ optimal.
\end{remark}

Finally, we show that if $\mu^*$ is sufficiently large, then it is optimal to attempt the transmission of all packets in the queue.

\begin{proposition}
\label{prop:mx}
If $\mu^*\geq \min\{1, \frac{4d}{1+\lambda}\}$, then $x_1 = a$ where $(m_1, x_1)$ is the optimal code selection.
\end{proposition}

\begin{proof}
For any $x>0$, let $m = m(x) \in M^*(x)$ and $q(x) = \nu(m(x), x)$. We will show that $dq(x)/dx \geq 0$ for all $x$, which says that $x_1 = a$ and $m_1 \in M^*(a)$. By Envelope Theorem \cite{EnvTh},
\begin{eqnarray*}
	\frac{dq(x)}{dx} &=& \frac{\partial}{\partial x}\nu(m, x), \\
	&=& (1+\lambda)(1-\Phi\Big(\frac{x-m\mu^*}{\sqrt{m\mu^*(1-\mu^*)}}\Big)) \\ && -(1+\lambda)\frac{x}{\sqrt{m\mu^*(1-\mu^*)}}\varphi\Big(\frac{x-m\mu^*}{\sqrt{m\mu^*(1-\mu^*)}}\Big), \\
	&=& (1+\lambda)(1-\Phi\Big(\frac{x-m\mu^*}{\sqrt{m\mu^*(1-\mu^*)}}\Big))\\ && \hskip 1.75in -\frac{2d}{x/m+\mu^*},
\end{eqnarray*}
\noindent where the last equality follows from the optimality condition in Lemma \ref{lemma:opt}. By Proposition \ref{prop:rate}, we have $x < m\mu^*$, which indicates that $(1-\Phi\Big(\frac{x-m\mu^*}{\sqrt{m\mu^*(1-\mu^*)}}\Big)) > 1/2$. Therefore,
\begin{eqnarray*}
\frac{dq(x)}{dx} &>& \frac{1+\lambda}{2} - \frac{2d}{\mu^*}, \\
 &\geq& 0,
\end{eqnarray*}
\noindent where the last line is true if $\mu^* \geq \min\{1, \frac{4d}{1+\lambda}\}$. Thus, $dq(x)/dx > 0$ and the proof follows.
\end{proof}

\noindent This result is particularly important as it eliminates one of the constraints in (\ref{eqn:dp}).
  
\textbf{Example:} In Figure \ref{fig:opt_intol_mx} and Figure \ref{fig:opt_intol_r}, we investigate the behavior of $(m_1,x_1)$ and code rate $r$ over all possible $\mu^*\in[0,1]$ in the delay-intolerant setting under the continuous approximation. In this example, we assume that $a = 6$, $d = 0.25$, $\lambda = 1$.
 \begin{figure}[htb] 
 \centering
\includegraphics[width=.7\linewidth]{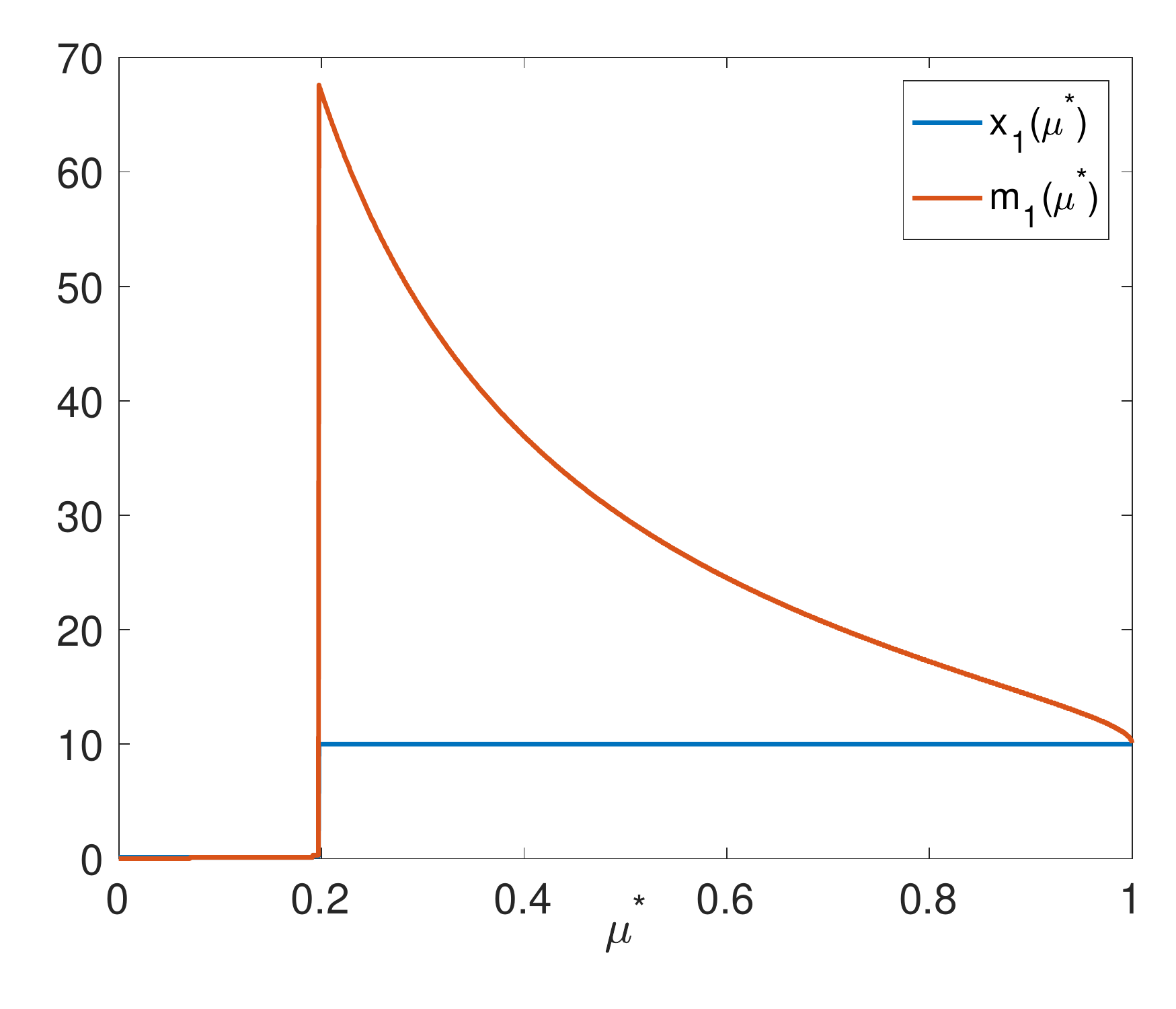}   
\caption{The optimal $(m_1,x_1)$ for $\mu^*\in[0,1]$. It is optimal to transmit all packets in the queue for large enough $\mu^*$.}
 \label{fig:opt_intol_mx}
\end{figure}
 \begin{figure}[htb] 
  \centering
{\includegraphics[width=.7\linewidth]{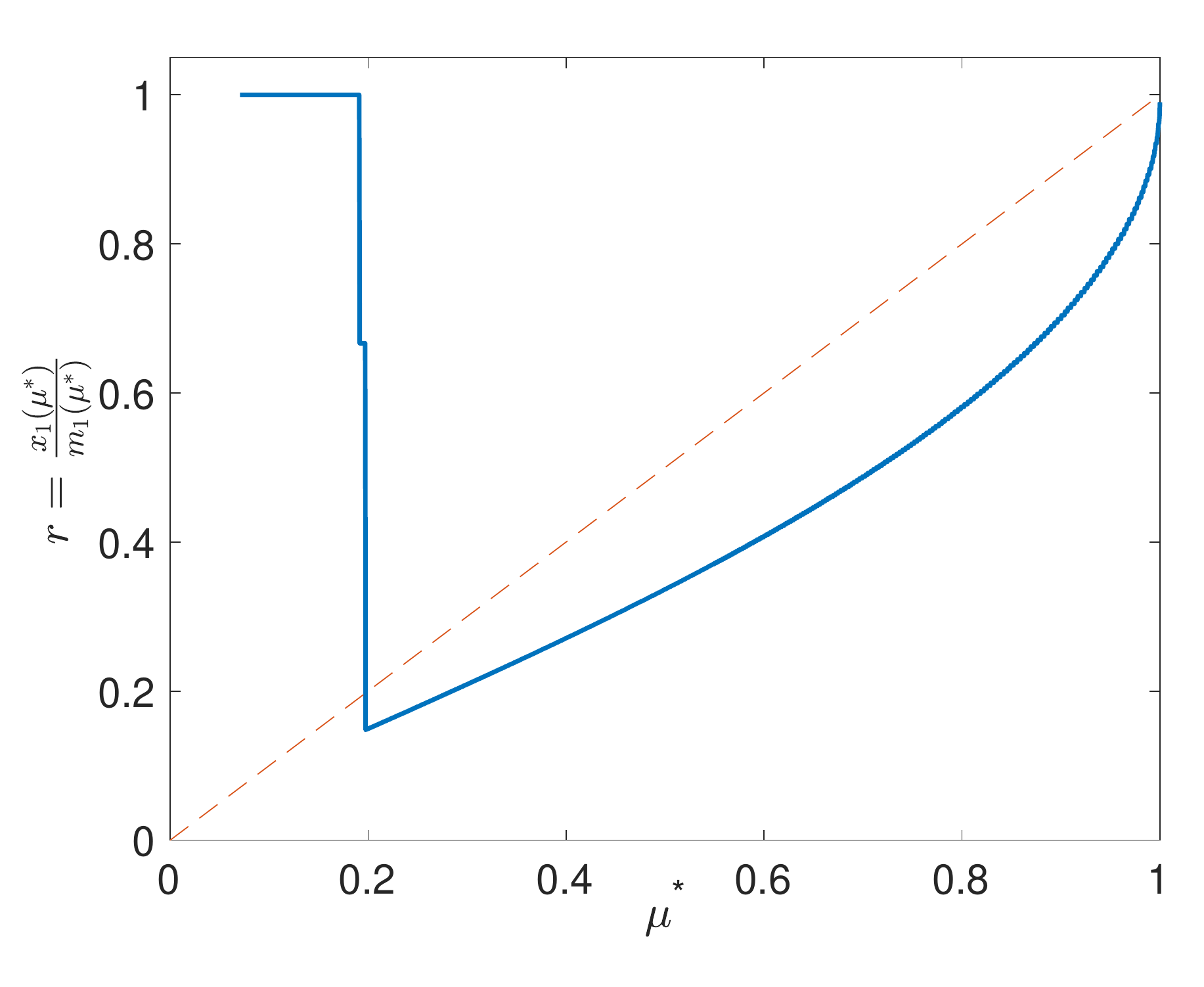} }
\caption{The optimal code rate $r$ for $\mu^*\in[0,1]$. Code rate is strictly below $\mu^*$ for $\mu^*\geq \frac{2d}{1+\lambda}$.}
 \label{fig:opt_intol_r}
\end{figure}

\noindent In this case, we first observe that $r < \mu^*$ for $\mu^*\geq \frac{2d}{1+\lambda}$ and $x_1 = a$ for $\mu^*\geq \frac{4d}{1+\lambda}$, which verify Proposition \ref{prop:rate} and Proposition \ref{prop:mx}, respectively. Secondly, similar to the delay-tolerant case, the channel use becomes less favorable as the reliability $\mu^*$ increases beyond a certain level since an additional channel is not required when the existing ones are already reliable.


In the following, we will introduce an index-based learning algorithm based on the optimal rate allocation algorithm, and show that it achieves desirable performance.

\section{Learning Algorithm: UCB-Deadline}
%

In this section, we will introduce an index-based algorithm called that UCB-Deadline, which achieves order-optimal regret performance for the exploration-and-exploitation problem at hand. In this section, we first develop the algorithm, then find regret upper bounds that will lead us to the optimality result.

%
%
%

\subsection{UCB-Deadline}

For the exploration-and-exploration problem at hand, learning must be reinforced when the confidence is low in order to avoid linear regret in certain sample paths on which exploration is stopped at an early stage, and the estimates must converge to the true mean after a sufficiently long time for achieving small regret in the long-run. Utilization of upper confidence bound (UCB) in the absence of the true mean reinforces learning through "optimism in the face of uncertainty" \cite{Bubeck2012}, therefore is a suitable strategy in algorithm design. In the following, we define a policy named UCB-Deadline that makes use of UCB to determine the number of channels to be activated.
\begin{definition}[UCB-Deadline]
Let $M(n) = \sum_{t=1}^T m_t$ be the number of channels that are activated in frame $n$, $Z(n) = \sum_{k=1}^n M(k)$ be the number of activated channels until frame $n$, and
\begin{equation}
\xi(n)=\frac{1}{Z(n)}\sum\limits_{k=1}^n \sum_{t=1}^T \sum\limits_{i\in\mathcal{U}_k(t)}C_i^k(t),
\end{equation}
be the sample mean of activated channels until frame $n$, and $c_{n, s}=\sqrt{\frac{\beta\log n}{2s}}$ for $\beta > 0$. UCB at frame $n$ is defined as follows:
\begin{equation}
	\bar{\mu}_{Z(n-1)}(n)=\xi(n-1)+c_{n,Z(n-1)}.
\end{equation}

\noindent Let $\hat{\mu}_s$ be the sample mean of channel realizations after $s$ channel uses. Since all channels are iid and symmetric, $\xi(n-1) \overset{d}{=} \hat{\mu}_{Z(n-1)}$, which will provide simplicity in the performance analysis.

With these definitions, UCB-Deadline with parameter $\beta$, denoted as UCB-Deadline($\beta$), is summarized in Algorithm \ref{alg:UCB},
\begin{algorithm}
\SetKwData{Left}{left}\SetKwData{This}{this}\SetKwData{Up}{up}
\SetKwFunction{Union}{Union}\SetKwFunction{FindCompress}{FindCompress}
\SetKwInOut{Input}{input}
\Input{$\beta>0$}
\BlankLine

\textbf{Initialization}:
$Z(0)=1$;~
$\xi(0)\sim Ber(\mu^*)$;

\For{$n=1,2,\ldots,N$}{
	$\bar{\mu}_{Z(n-1)}(n) = \xi(n-1)+c_{n, Z(n-1)}$; \\
	$X_T = A(n)$;\\
	\For{$t=T,T-1,\ldots, 1$}{
		$(m_t,x_t) = (\pi_{A(n)}^{*})^t(\bar{\mu}_{Z(n-1)}(n))$;\\
		\eIf{$m_t \geq x_t$} {$X_{t-1} = X_t-x_t$;} {$X_{t-1}=X_t$;} 
		}

		$Z(n) = Z(n-1)+\sum_{t=1}^Tm_t$;\\
		$\xi(n)=\frac{1}{Z(n)}\cdot\Big(Z(n-1)\cdot \xi(n-1)+\sum_{t=1}^T\sum\limits_{i=1}^{m_t}C_i^n(t)\Big)$;
}
\caption{UCB-Deadline($\beta$)}
\label{alg:UCB}
\end{algorithm}

\noindent where $(\pi_{A(n)}^*)^t$ is the optimal policy defined in (\ref{eqn:opt_rate_alloc}).
\end{definition}

In the following subsection, performance guarantees under UCB-Deadline will be presented in the form of regret upper bounds.

\subsection{Regret Analysis of UCB-Deadline}
\label{sec:perf}

In this section, we will provide upper bounds for the regret under UCB-Deadline. The strategy to accomplish this is as follows: first we will provide two lemmas in a general setting, and then use these lemmas to upper bound the regret under UCB-Deadline.

\begin{lemma}
\label{lemma:01}
Consider a case where the optimal policy is $\pi^*_{a}(\mu^*)$ makes decisions $\mathbf{m}_a\cdot\mathbb{I}_{\{\mu^* > \zeta\}}$ where $\mathbf{m}_a = (m^T_a,\ldots,m^1_a)$. Let $Z_0(n) = \sum_{t=1}^n \mathbb{I}_{\{\pi_{A(n)}^*(\bar{\mu}_{Z(n-1)}(n)) = (\mathbf{0},\mathbf{0})\}}$ be the number of frames when all channels are idle under UCB-Deadline. Under UCB-Deadline with $\beta \geq 3$, the following upper bounds are obtained for any $a\in\mathcal{A}$: 
\begin{enumerate}
	\item If $\mu^* > \zeta$, then $\mathbb{E}[Z_0(N)] \leq TM_{max}\frac{\pi^2}{6}$,
	\item If $\mu^* \leq \zeta$, then 
	\begin{equation*} \mathbb{E}[N-Z_0(N)] \leq \frac{2\beta\log N}{(\zeta-\mu^*)^2}+TM_{max}\frac{\pi^2}{6},			\end{equation*}
\end{enumerate}
\noindent for all $N\geq 1$.
\end{lemma}

\noindent Lemma \ref{lemma:01} implies that in a binary decision case, UCB-Deadline makes a bounded number of wrong decisions if the true mean is higher than the critical point, and a logarithmically growing number of wrong decisions over time otherwise in the expected sense.

\begin{lemma}
	\label{lemma:012}
	Fix $a\in\mathcal{A}$. Let $\zeta^l_a < \zeta^u_a$ be two given constants in $[0, 1]$. Consider the following optimal policy:
\begin{equation}
    (m_t, x_t) = 
\begin{cases}
    (0, 0),& \text{if } \mu^* < \zeta^l_a\\
    (m_t^*, x_t^*),              & \text{if } \mu^* \in [\zeta^l_a,\zeta^u_a] \\
    (\tilde{m}_t, \tilde{x}_t), & \text{if } \mu^* > \zeta^u_a.
\end{cases}
\label{eqn:lemma012_opt}
\end{equation}
for some $(m_t^*,x_t^*)>0$ and $(\tilde{m}_t,\tilde{x}_t)>0$. Assume $\mu^* \in [\zeta^l_a,\zeta^u_a]$. Under UCB-Deadline with $\beta \geq 4$, the following upper bounds hold for all $n\geq 1$:
\begin{enumerate}
	\item $\mathbb{E}[ \sum_{n=1}^N \mathbb{I}_{\{\pi_{A(n)}^*(\bar{\mu}_{Z(n-1)}(n)) = (0,0)\}}] \leq TM_{max}\frac{\pi^2}{6}$.
	\item $\mathbb{E}[ \sum_{n=1}^N \mathbb{I}_{\{\pi_{A(n)}^*(\bar{\mu}_{Z(n-1)}(n)) = (\tilde{m},\tilde{x})\}}]   \leq TM_{max}\frac{\pi^2}{6}+\Psi(\frac{(\mu^*-\zeta_a^u)^2}{2\beta}) < \infty$, where
	\begin{equation*}
		\Psi(\epsilon) = n_\epsilon+TM_{max}\cdot\frac{\pi^2}{2}\sum_{n=1}^\infty\frac{1}{(n-\log(n+1)/\epsilon)^2}
\end{equation*}
\noindent and $n_\epsilon=\inf\{n:n-\frac{\log(n+1)}{\epsilon}>0\}$.
\end{enumerate}
\end{lemma}

\noindent Lemma \ref{lemma:012} says that if the true mean is in an interval with nonempty interior so that the correct decision can be made after sufficient concentration around the mean, then the numbers of wrong decisions under UCB-Deadline are bounded in both directions in the expected sense.

Proofs of Lemma \ref{lemma:01} and Lemma \ref{lemma:012} will be given in Appendix.

The following theorem provides performance guarantees under UCB-Deadline.

\begin{theorem}[Regret Upper Bounds for UCB-Deadline]
The following upper bounds hold for the regret under UCB-Deadline with parameter $\beta\geq 4$.
\begin{enumerate}
	\item If $\mu^* < \zeta$, then
	\begin{equation}
		\bar{R}_N \leq \sum\limits_{a\in\mathcal{A}}\alpha_a \cdot B_a\cdot\Big (\frac{2\beta\log N }{(d-\mu^*)^2}+TM_{max}\frac{\pi^2}{6}\Big ).
	\end{equation}
	
	\item If $\mu^* \geq \zeta$, let $\zeta_a^l\leq \mu^* \leq \zeta_a^u$ be the largest interval such that $\pi_a^*(\mu)=\pi_a^*(\mu^*),~\forall{\mu}\in[\zeta_a^l,\zeta_a^u]$ for any $a\in\mathcal{A}$. Then,
	\begin{multline}
		\bar{R}_N \leq \sum\limits_{a\in\mathcal{A}}\alpha_a \Big(\mathbb{E}[J(a,\pi_a^*(\mu^*))]+B_a\Big) \\ \cdot (TM_{max}\frac{\pi^2}{3}+\Psi\Big(\frac{(\mu^*-\zeta_a^u)^2}{2\beta}\Big)),
	\end{multline}
	\noindent for all $N\geq 1$.
\end{enumerate}

\label{thm:regret_bnd}
\end{theorem}

\begin{proof}
	\begin{enumerate}
		\item Note that $\mu^*<\zeta$ implies that $\pi_a^*(\mu^*)=(\mathbf{0},\mathbf{0})$ and therefore $J(a, \pi_a^*(\mu^*))=0,~\forall a\in\mathcal{A}$. Thus, the regret is upper bounded by using (\ref{eqn:regret}) as follows:
		\begin{eqnarray*}
		\bar{R}_N &=& \sum_{a\in\mathcal{A}} \alpha_a \cdot \mathbb{E} \Big[\sum_{n=1}^N -J(a, \pi_a^*(\bar{\mu}(n)))\Big] \\
		&\leq& \sum_{a\in\mathcal{A}} \alpha_a \cdot \mathbb{E} [\sum_{n=1}^N B_a\cdot\mathbb{I}_{\{I_a(t)\neq 0\}}] \\
		&=& \sum_{a\in\mathcal{A}} \alpha_a B_a\cdot \mathbb{E} [N-Z_0(N)] \\
		&\overset{(a)}{\leq}& \sum_{a\in\mathcal{A}} \alpha_a B_a \cdot(\frac{2\beta\log N}{(\zeta-\mu^*)^2}+TM_{max}\frac{\pi^2}{6}),
		\end{eqnarray*}
		\noindent  where (a) follows from Lemma \ref{lemma:01}.
		
		\item Let $\Delta J_{a}^{max} = \mathbb{E}[ J(a, \pi_a^*(\mu^*))+B_a]$. Note that $\Delta J_{a}^{max}$ is an upper bound for the instantaneous regret for any $a\in\mathcal{A}$. Then, the regret under UCB-Deadline can be upper bounded as follows:
		\begin{eqnarray*}
			\bar{R}_N \hskip -0.1in &\leq& \hskip -0.1in \sum_{a\in\mathcal{A}} \alpha_a \Delta J_{a}^{max} \sum_{n=1}^N \mathbb{E} \Big[   \mathbb{I}_{ \{ \pi_a^*(\bar{\mu}(n)) \neq \pi_a^*(\mu^*) \}} \Big] \\
			&\leq& \hskip -0.1in \sum_{a\in\mathcal{A}} \alpha_a \Delta J_{a}^{max} \sum_{n=1}^N \mathbb{E} \Big[   \mathbb{I}_{ \{ \bar{\mu}(n) < \zeta_a^l \}} + \mathbb{I}_{ \{ \bar{\mu}(n) > \zeta_a^u \}}  \Big] \\			
			&\overset{(a)}{\leq}& \hskip -0.1in \sum_{a\in\mathcal{A}} \alpha_a \Delta J_{a}^{max} \\ && \hskip 0.15in \cdot \sum_{n=1}^N \mathbb{E} \Big[   \mathbb{I}_{ \{ \pi_a^*(n) = 0 \}} + \mathbb{I}_{ \{ \pi_a^*(n) = \underset{\hat{\mu}>\zeta_a^u}{\min} ~I_a^*(\hat{\mu}) \}}  \Big] \\			
			&\overset{(b)}{\leq}& \hskip -0.1in \sum_{a\in\mathcal{A}} \alpha_a \Delta J_{a}^{max} (TM_{max}\frac{\pi^2}{3} + \Psi(\frac{(\zeta_a^u-\mu^*)^2}{2\beta}) ),
		\end{eqnarray*}
		where (a) follows from the fact that minimal learning and maximal possible regret per timeslot maximize the overall regret, and (b) is a direct application of Lemma \ref{lemma:012}.
	\end{enumerate}
\end{proof}

\noindent Theorem \ref{thm:regret_bnd} implies that the regret under UCB-Deadline is bounded if transmission is feasible, i.e., $\mu^* \geq \zeta$ where $\zeta$ is the critical point. This is an interesting result since in most exploration-exploitation problems, the regret is logarithmic \cite{Bubeck2012, Tekin}.

In the following theorem, we will state that UCB-Deadline is order-optimal in all cases.

\begin{theorem}[Optimality of UCB-Deadline]
For the learning problem, UCB-Deadline with parameter $\beta \geq 4$ is order optimal, i.e., no other admissible learning algorithm can achieve better than $\Theta(\log N)$ regret if $\mu^*<\zeta$ and $\Theta(1)$ regret if $\mu^*\geq \zeta$.
\end{theorem}

\begin{proof}
It is clear that $\Theta(1)$ is the best any policy can achieve if $\mu^*\geq \zeta$. If $\mu^*<\zeta$, the optimal policy turns out to be $(m_t,x_t)=(0,0),~\forall t$. This case is a straightforward instance of a classical stochastic multi-armed bandit scenario with two arms: $(m_t,x_t)=(0,0),~\forall t$ corresponds to pulling a hypothetical arm with higher yield, and any other decision corresponds to pulling a suboptimal arm. Since no channel is activated if $(m_t,x_t)=(0,0),~\forall t$, it is clearly a two-arm classical MAB problem. It is well-known that the regret is logarithmic in such cases, hence UCB-Deadline is order optimal in all cases.
\end{proof}

\section{Numerical Results}
\label{sec:num_res}
In this section, we will provide simulation results for performance analysis in a variety of scenarios. As a performance benchmark, we will utilize a Bayesian learning policy based on Thompson Sampling. We first introduce the learning policy.

In problems that involve exploration-and-exploitation tradeoff, Thompson Sampling provides effective solutions that reinforce learning through randomization \cite{Gopalan}. In the following, we propose an algorithmic prescription to the learning problem at hand based on Thompson Sampling, which is abbreviated as TS-Deadline.

\begin{definition}[TS-Deadline] Let $Beta(\theta_0, \theta_1)$ denote the beta distribution with parameters $\theta_i > 0$ for $i=0,1$ whose probability density function is given by $f(x;\theta_0,\theta_1)=\frac{\Gamma(\theta_0+\theta_1)}{\Gamma(\theta_0)\Gamma(\theta_1)}(1-x)^{\theta_0-1}x^{\theta_1-1}$ \cite{Gopalan}. TS-Deadline is described in Algorithm \ref{alg:TS}.

\begin{algorithm}
\SetKwData{Left}{left}\SetKwData{This}{this}\SetKwData{Up}{up}
\SetKwFunction{Union}{Union}\SetKwFunction{FindCompress}{FindCompress}
\BlankLine

\textbf{Initialization}:
$\theta_0(0)=1,~\theta_1(0)=1$

\For{$n=1,2,\ldots,N$}{
	$\bar{\mu}^{TS}(n)\sim Beta(\theta_0(n-1),\theta_1(n-1))$; \\
	$X_T = A(n)$;\\
	\For{$t=T,T-1,\ldots, 1$}{
		$(m_t,x_t) = (\pi_{A(n)}^{*})^t(\bar{\mu}^{TS}_{Z(n-1)}(n))$;\\
		\eIf{$m_t \geq x_t$} {$X_{t-1} = X_t-x_t$;} {$X_{t-1}=X_t$;} 
		}

		$\theta_k(n)=\theta_k(n-1)+\sum_{t=1}^T\sum\limits_{i=1}^{m_t}\mathbb{I}_{\{C_i^n(t)=k\}},~k=0,1$.
}
\caption{TS-Deadline}
\label{alg:TS}
\end{algorithm}

\end{definition}

\subsection{Regret Investigation in Extreme Cases}
We first analyze the performance of UCB-Deadline and compare with TS-Deadline for the delay-tolerant and delay-intolerant systems investigated in Section \ref{sec:prop}.

\subsubsection{Performance in Delay-Tolerant Scenario}
We analyze the performance of the learning algorithms in the delay-tolerant system for which we characterized the optimal rate allocation policy in Section \ref{sec:prop}. For $A(n) = 1,~\forall n$, $d = 0.25$, $T=4$, $\lambda = 1$, recall that the optimal rate allocation algorithm is illustrated in Figure \ref{fig:dl_tol}. For this specific example, we choose two $\mu^*$ values, which are below and above the critical point, and analyze the performances of UCB-Deadline and TS-Deadline.

First, we consider $\mu^* = 0.05$, which is below the critical point as it can be seen in Figure \ref{fig:dl_tol}. For this case, $\mathbf{m} = (0,0,0,0)$ and $J_T(1) = -1$ given the true mean $\mu^*$. The regrets and throughputs under UCB-Deadline and TS-Deadline are given in Figure \ref{fig:dl_tol_inf_r}-\ref{fig:dl_tol_inf_thru}.


\begin{figure}
\centering
 \includegraphics[width=.7\linewidth]{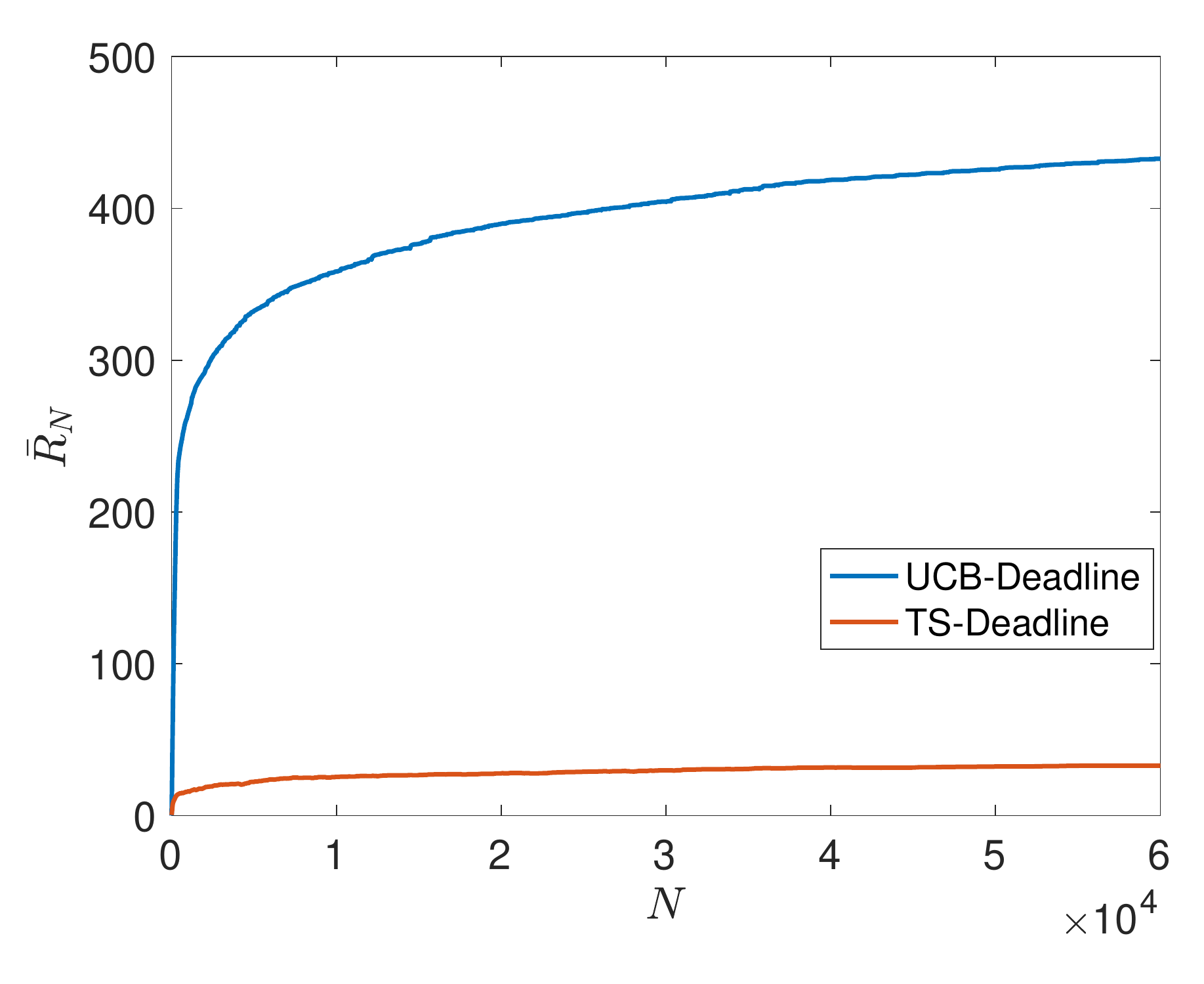}
 \caption{Regrets achieved by UCB-Deadline and TS-Deadline for the case $A(n) = 1,~\forall n$, $d = 0.25$, $T=4$, $\lambda = 1$ and $\mu^* = 0.05$. Both algorithms achieve logarithmic regret, and TS-Deadline provides faster convergence compared to UCB-Deadline up to a coefficient.}
 \label{fig:dl_tol_inf_r}
\end{figure}

\noindent By Theorem \ref{thm:regret_bnd}, the upper bound for the regret under UCB-Deadline is logarithmic over time, consistent with these simulation results. It is observed that TS-Deadline also has an increasing regret, but it achieves lower regret than UCB-Deadline in this case.
\begin{figure}
\centering
 \includegraphics[width=.7\linewidth]{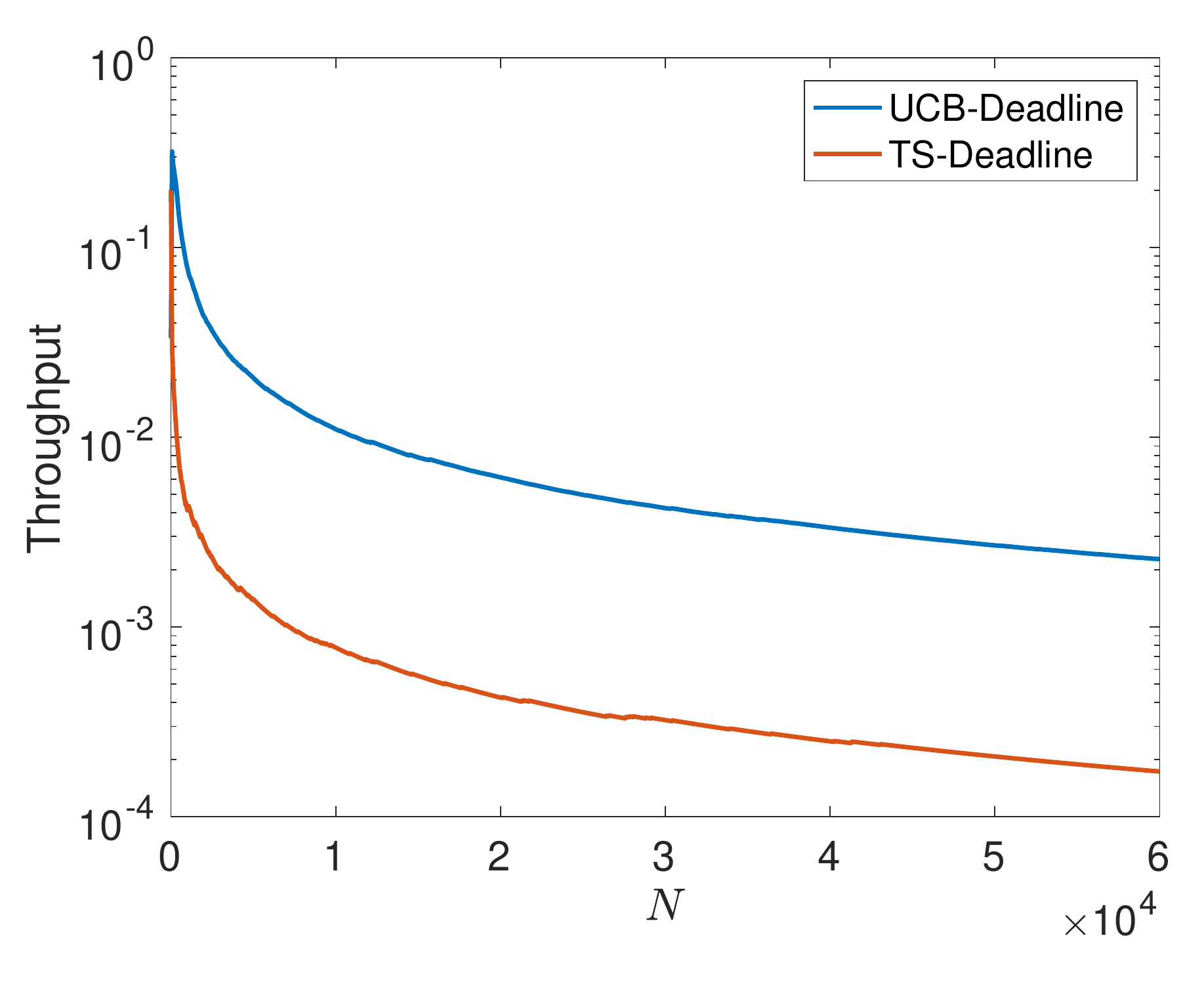} 
\caption{Throughputs achieved by UCB-Deadline and TS-Deadline for the case $A(n) = 1,~\forall n$, $d = 0.25$, $T=4$, $\lambda = 1$ and $\mu^* = 0.05$. Both algorithms achieve logarithmic regret, and TS-Deadline provides faster convergence compared to UCB-Deadline up to a coefficient.}
\label{fig:dl_tol_inf_thru}
\end{figure}

\noindent Since $\mu^* < \zeta$ in this case, it is optimal to stay idle, which implies zero throughput. From Figure \ref{fig:dl_tol_inf_thru}, we observe that throughput under both algorithms decay, the decay rate is higher under TS-Deadline.

In order to observe the behavior of the learning policies above the critical point, we consider $\mu^* = 0.7$ in the same delay-tolerant setting. The regret performances of UCB-Deadline and TS-Deadline are given in Figure \ref{fig:dl_tol_f_r}.

\begin{figure}
\centering
 \includegraphics[width=.8\linewidth]{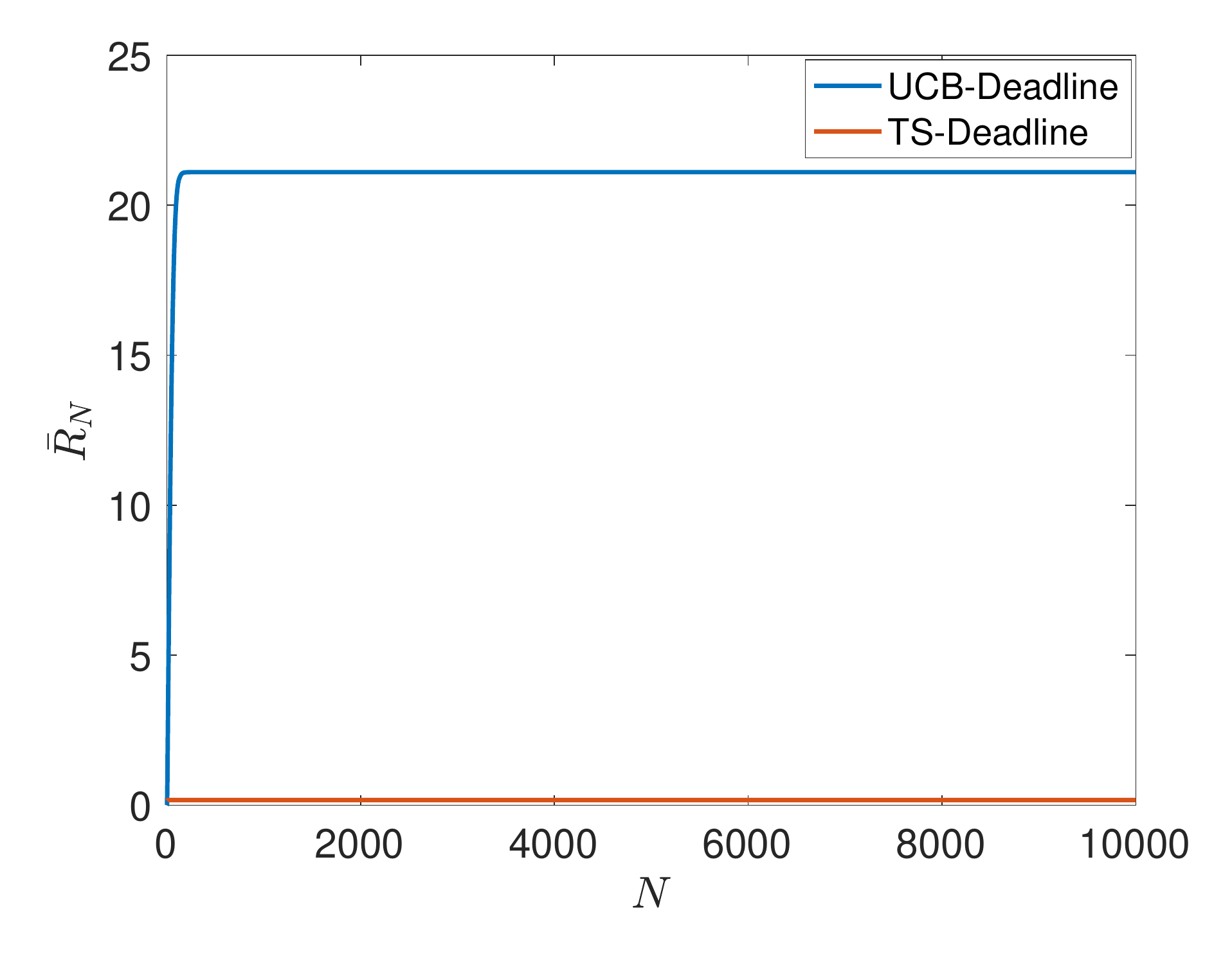} 
\caption{The regret performances of UCB-Deadline and TS-Deadline for the case $A(n) = 1,~\forall n$, $d = 0.25$, $T=4$, $\lambda = 1$ and $\mu^* = 0.7$. Both algorithms achieve bounded regret.}
\label{fig:dl_tol_f_r}
\end{figure}

\noindent From Figure \ref{fig:dl_tol_f_r}, it is observed that the regret stays bounded, which is foreseen by Theorem \ref{thm:regret_bnd}. 

 \begin{figure}[htb] 
\centering
\includegraphics[width=.7\linewidth]{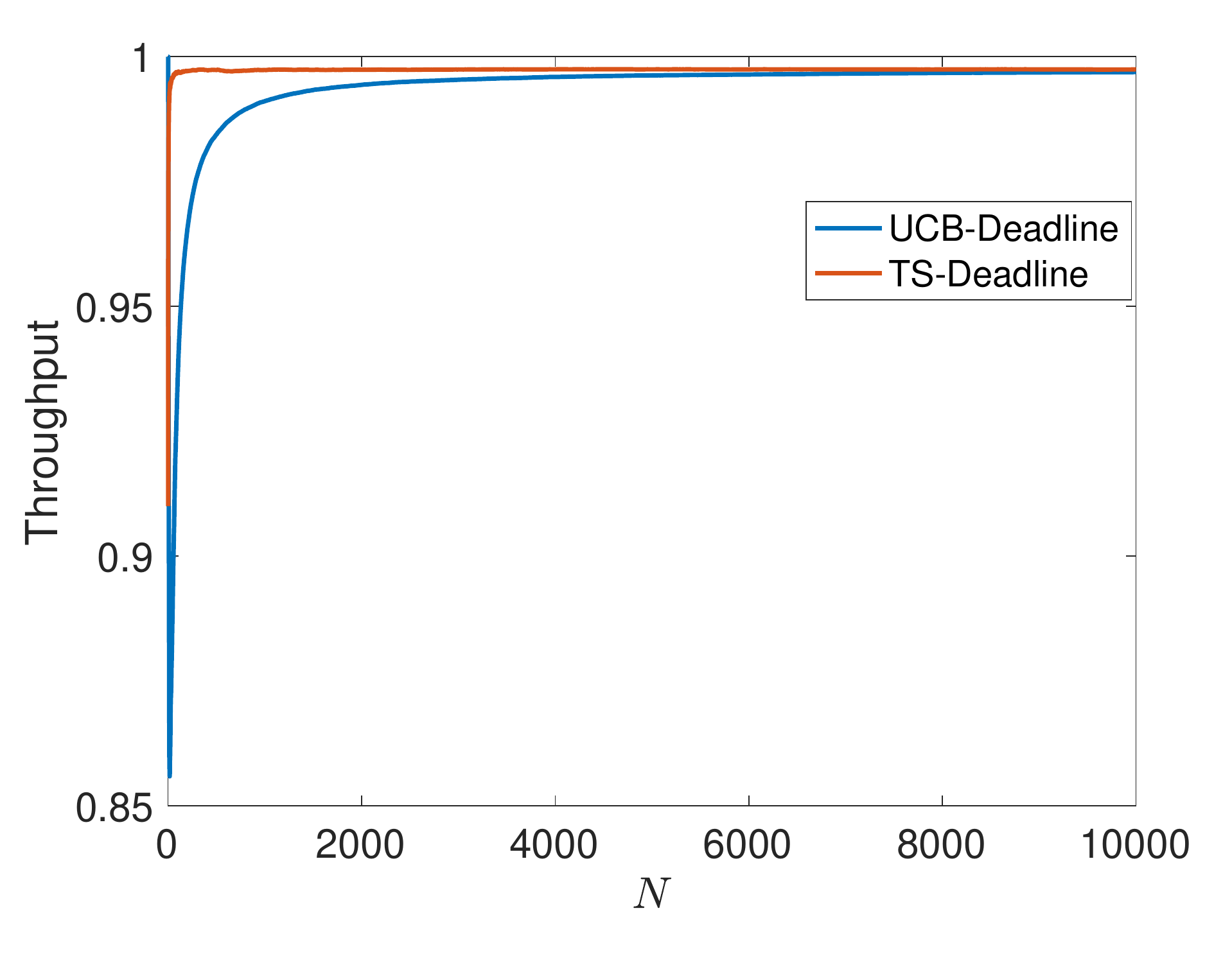}
\caption{Throughputs under UCB-Deadline and TS-Deadline for the case $A(n) = 1,~\forall n$, $d = 0.25$, $T=4$, $\lambda = 1$ and $\mu^* = 0.7$. Throughput converges to 1 in both cases.}
\label{fig:dl_tol_f_thru}
\end{figure}

\noindent Throughput under both algorithms converge to $1$. TS-Deadline provides faster convergence up to a coefficient in this case as well.

\subsection{Performance in Delay-Intolerant Scenario}
In this subsection, we analyze the performance of UCB-Deadline and TS-Deadline in a delay-intolerant scenario. In these simulations, we assume that $T=1$, $\lambda = 0$, $d=0.2$ and $A_max=6$. The arrival process is chosen as an iid uniform distribution, which has the following probability mass function:
\begin{equation}
    \mathbb{P}(A(n) = k) = 
\begin{cases}
    \frac{1}{A_{max}+1},& \text{if } 0 \leq k \leq A_{max} \\
    0   , & \text{otherwise}
\end{cases}
\label{ex:poisson}
\end{equation}

\noindent where $A_{max}$ is the maximum number of arrivals in a frame.

Performance results for $\mu^*=0.05$ and $d=0.2$ are illustrated in Figure \ref{fig:R0_15} with the same truncated Poisson distribution for the arrival process. Note that $\mu^* < \zeta$ in this case, and therefore channel usage is infeasible for any queue-length. By Theorem \ref{thm:regret_bnd}, the upper bound for the regret under UCB-Deadline is logarithmic over time, consistent with the simulation results. It is observed that TS-Deadline also has an increasing regret, but it achieves significantly lower regret than UCB-Deadline in this case as well.

 \begin{figure}[htb] 
\centering
\includegraphics[width=.7\linewidth]{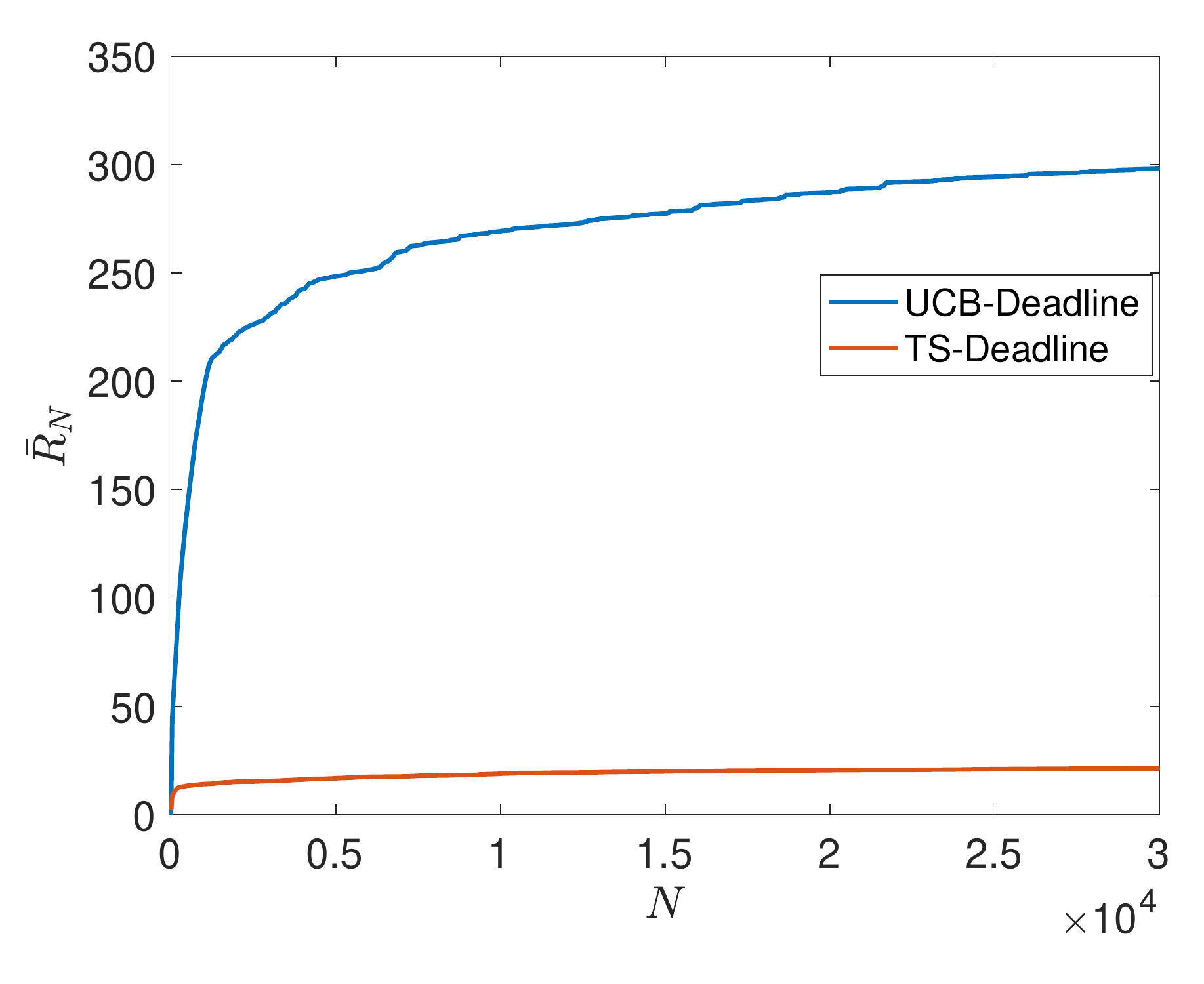}
\caption{The regret performances of UCB-Deadline and TS-Deadline for the case $A_{max} = 6$, $d = 0.25$, $T=1$, $\lambda = 1$ and $\mu^* = 0.05$. Both algorithms have logarithmic regret.}
\label{fig:R0_15}
\end{figure}

 \begin{figure}[htb] 
\centering
\includegraphics[width=.7\linewidth]{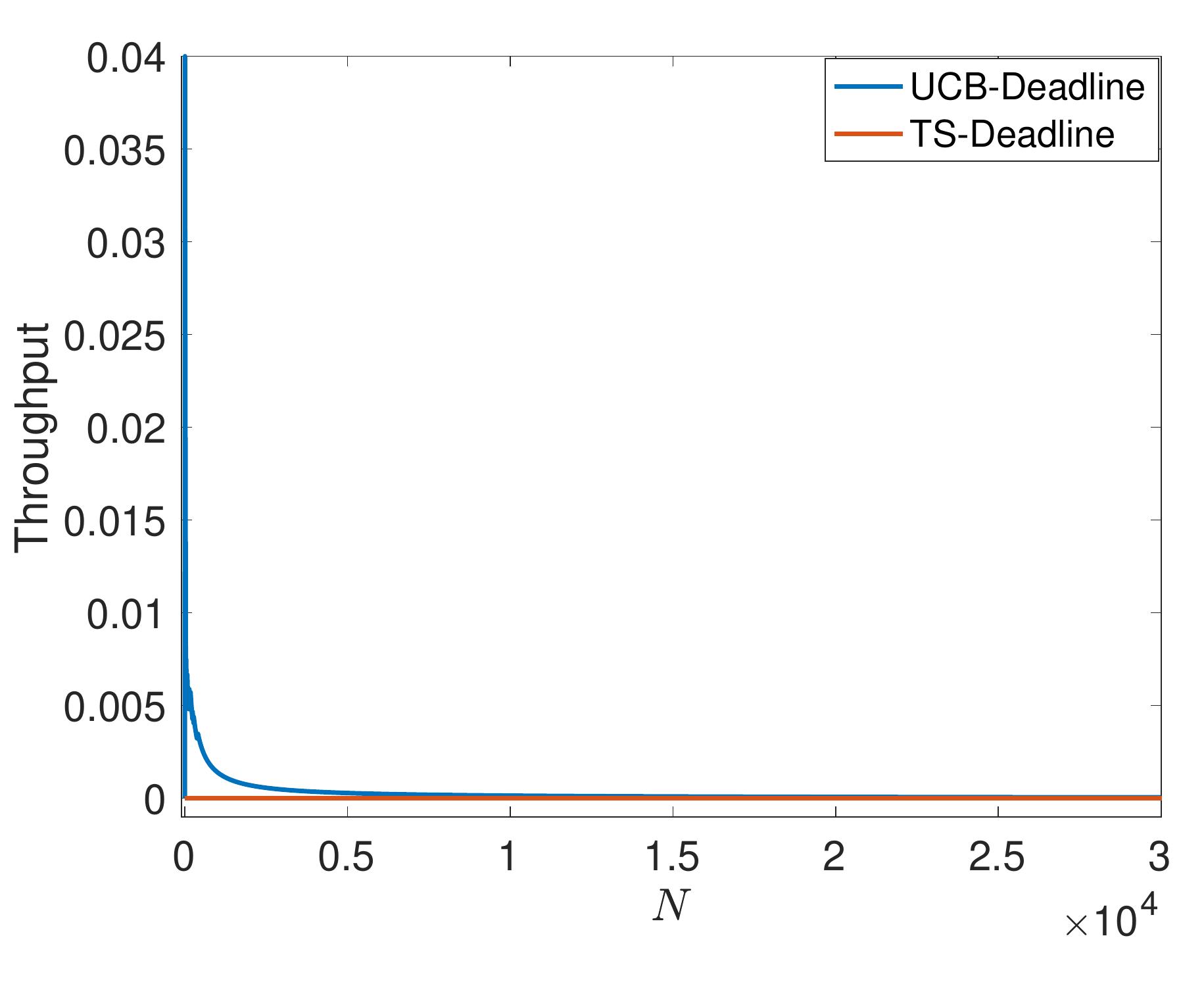}
\caption{The throughput performances of UCB-Deadline and TS-Deadline for the case $A_{max} = 6$, $d = 0.25$, $T=1$, $\lambda = 1$ and $\mu^* = 0.05$. Both algorithms have logarithmic regret.}
\label{fig:R0_15_thru}
\end{figure}


For $\mu^*=0.81$ and $d=0.25$, simulation results under UCB-Deadline and TS-Deadline are provided in Figure \ref{fig:R0_52}. The arrival distribution $\{\alpha_a\}_{a\in\mathcal{A}}$ is chosen as a truncated Poisson distribution with maximum element $A_{max} = 6$.

 \begin{figure}[htb] 
\centering
\includegraphics[width=.7\linewidth]{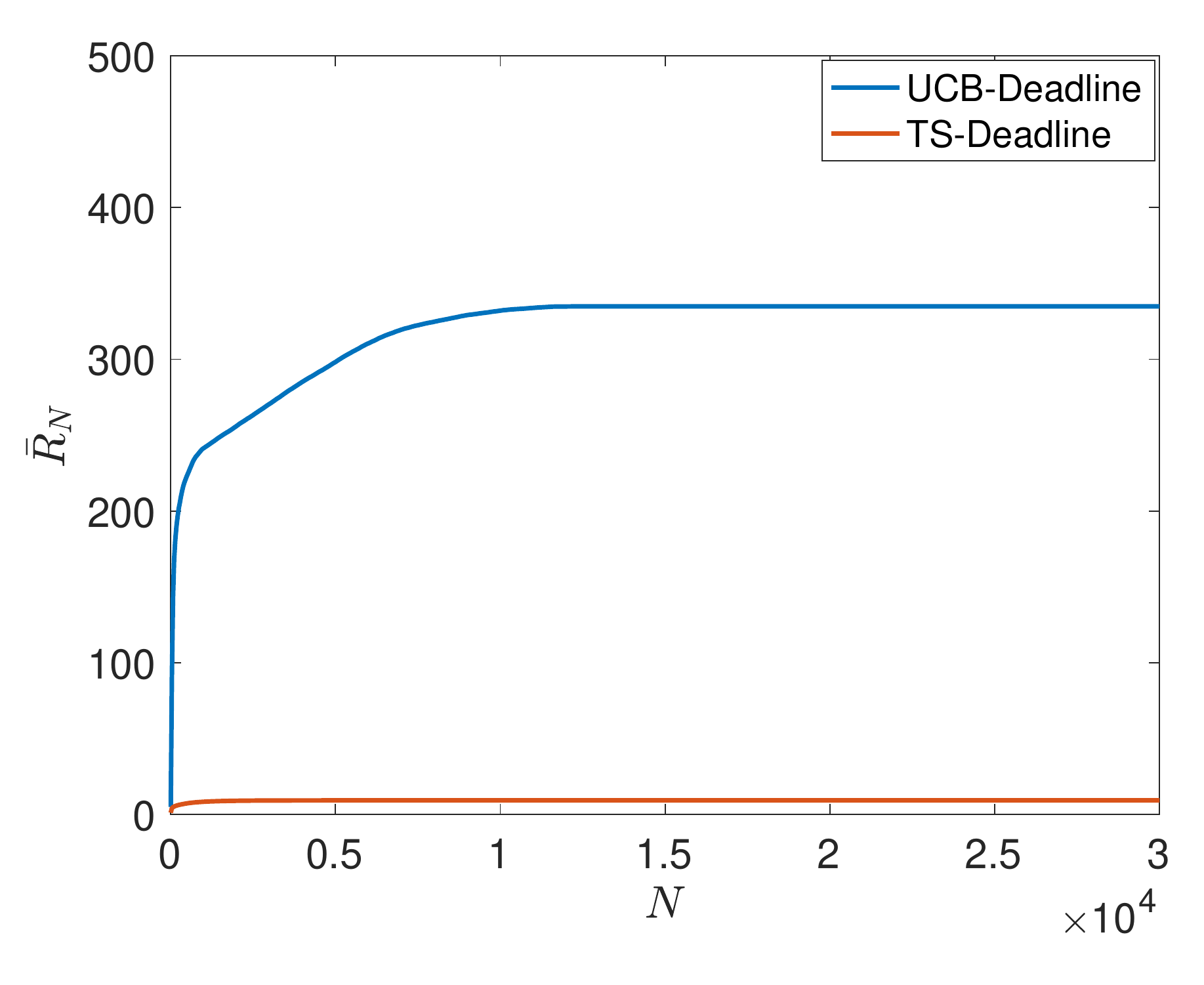}
\caption{The regret and throughput performances of UCB-Deadline and TS-Deadline for the case $A_{max} = 6$, $d = 0.25$, $T=1$, $\lambda = 1$ and $\mu^* = 0.81$. Both algorithms achieve bounded regret.}
\label{fig:R0_52}
\end{figure}

 \begin{figure}[htb] 
\centering
\includegraphics[width=.7\linewidth]{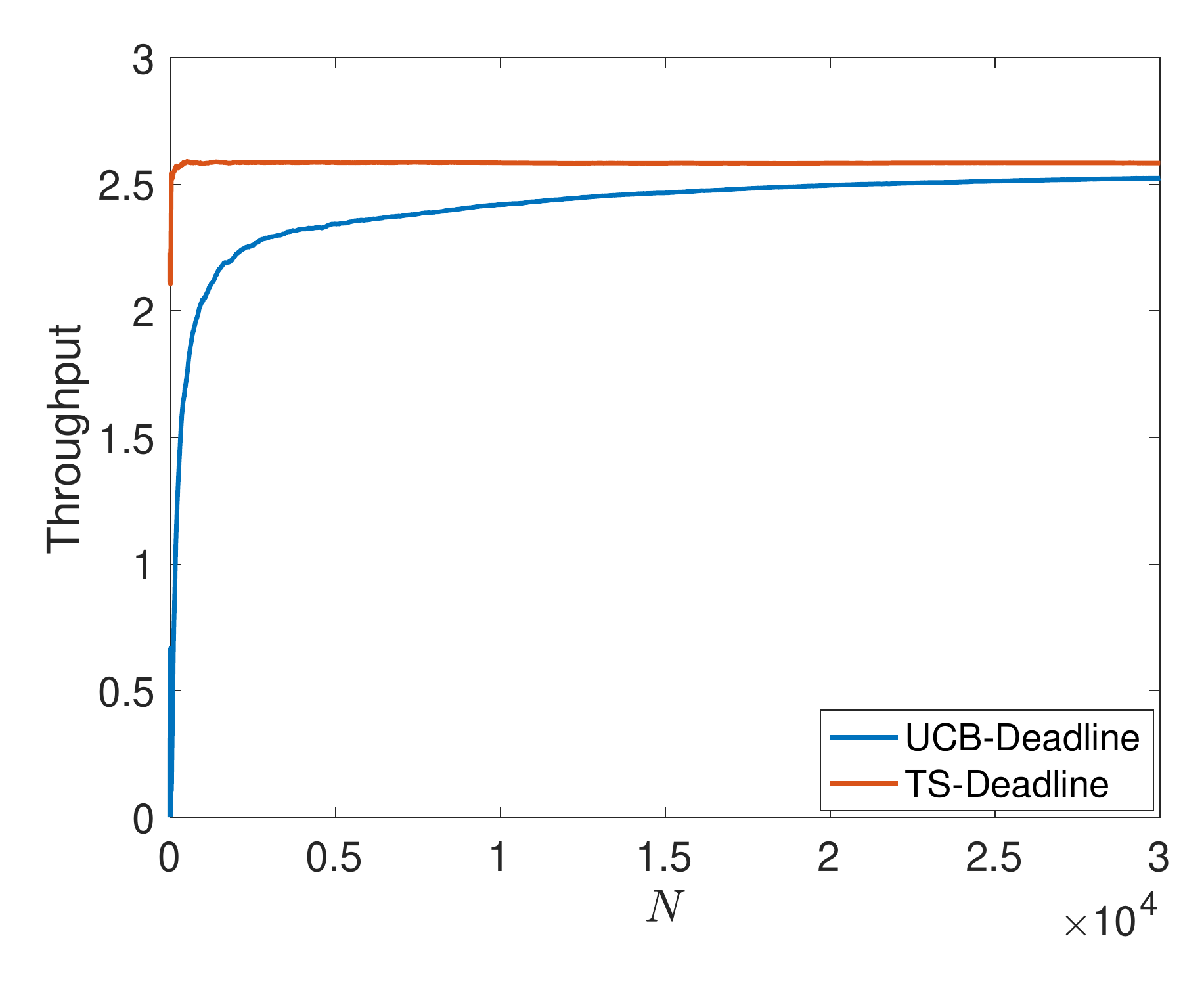}
\caption{The regret and throughput performances of UCB-Deadline and TS-Deadline for the case $A_{max} = 6$, $d = 0.25$, $T=1$, $\lambda = 1$ and $\mu^* = 0.81$. Both algorithms achieve bounded regret.}
\label{fig:R0_52_thru}
\end{figure}

\noindent Since $\mu^* > \zeta$ in this case, the regret is bounded by Theorem \ref{thm:regret_bnd}, which is verified by Figure \ref{fig:R0_52}. Also, it is noteworthy that TS-Deadline achieves smaller regret than UCB-Deadline in this case.

In both settings, we observed that TS-Deadline has a better regret performance compared to UCB-Deadline up to a coefficient. In the following, we present an example in which UCB-Deadline performs significantly better than TS-Deadline.

\subsection{Disadvantage of TS-Deadline}
As we saw in the previous cases, TS-Deadline provides lower regret than UCB-Deadline as in many other exploration-exploitation problems \cite{Gopalan} due to its fast convergence rate. In this subsection, we will provide an interesting case where UCB-Deadline outperforms TS-Deadline.

Consider a delay-intolerant case where $A(n) = 2$, $T=1$, $\lambda = 0$ and $d = 0.2$, and the number of channels is limited as $m_1\leq 2$. In this specific case, the performance of each algorithm for a horizon $N=10000$ is illustrated in Figure \ref{fig:comp2}.

 \begin{figure}[htb] 
\centering
\includegraphics[width=.7\linewidth]{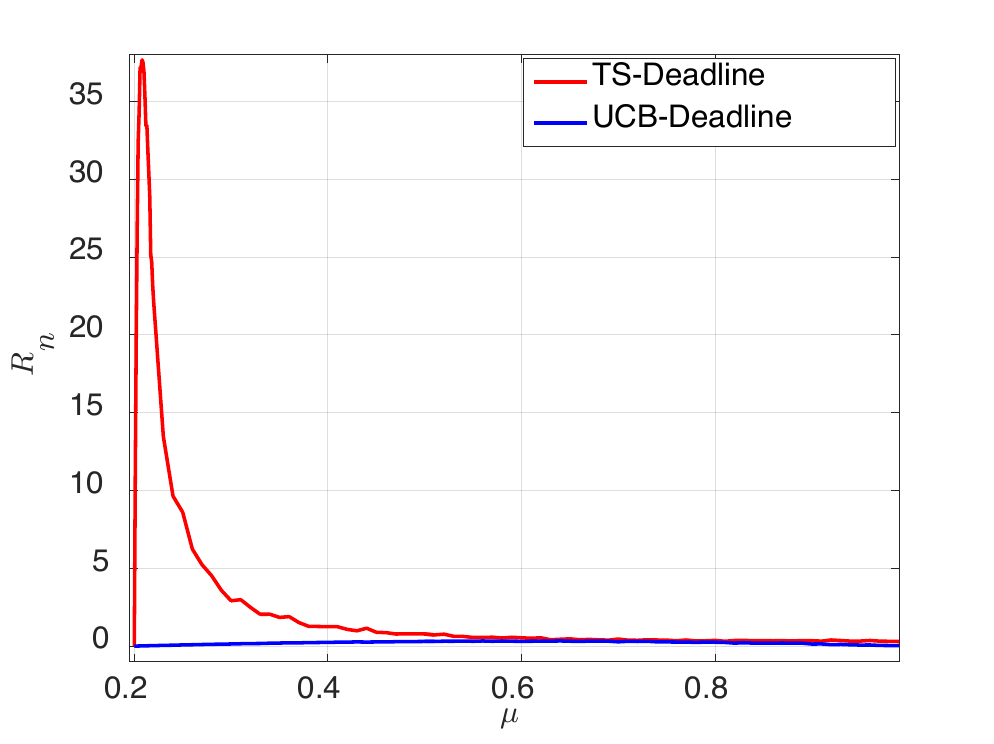}
\caption{The regret performances of UCB-Deadline and TS-Deadline for the case $A(n) = 2,~\forall n$, $d = 0.2$, $T=1$, $\lambda = 0$ and $N=10000$. UCB-Deadline outperforms TS-Deadline if $\mu^*>d$.}
\label{fig:comp2}
\end{figure}

\noindent In this example, we observe that if $\mu^*>d$, then UCB-Deadline outperforms TS-Deadline significantly. This is because UCB-Deadline has a positive bias which provides bigger advantage against TS-Deadline, which has a two-way bias due to the randomization when $\mu^*$ is slightly above the critical point. This particular structure enables UCB-Deadline to outperform TS-Deadline in this specific example.

\section{Conclusion}
In this paper, we investigated the channel allocation problem in a wireless network under a deadline-constrained traffic when the channel statistics and channel state information are unknown. We first identified the optimal rate allocation policy assuming that the channel statistics are known by the controller, and analyzed its important characteristics. Then, we proposed an index-based learning algorithm named UCB-Deadline. We proved that the regret under UCB-Deadline is bounded in the likely case that channel use is feasible, and logarithmic otherwise. This is an interesting result as the regret is logarithmic in most MAB problems.

It is assumed that there is a single class of independent and statistically symmetric channels in this work. UCB-Deadline is proved to achieve a bounded regret by incorporating the number of pending packets and utilizing the knowledge of statistical symmetry of the channels. In an extension of this setting where there are multiple classes of statistically symmetric channels, a similar exploitation of statistical symmetry may provide significant performance improvements. As a future work, we would like to investigate the learning problem in this extended setting.

As a performance benchmark, we introduced a Bayesian learning policy named TS-Deadline, and investigated its performance numerically. We observed that it achieves lower regret than its UCB counterpart in some cases, but there exist cases where UCB-Deadline outperforms TS-Deadline, i.e., TS-Deadline is not uniformly better than UCB-Deadline.

On the side of the service, an interesting extension of this work might be the learning problem where certain QoS requirements such as delivery ratio and service regularity must be met.

\appendix

\subsection{Proof of Proposition \ref{prop:cp}}
\begin{proof}
The proof stems from the following lemmas.
\begin{lemma}
\label{lemma:cont}
	If $m_1 = 0$, then $m_t = 0$ for all $t > 0$.
\end{lemma}
\begin{proof} From (\ref{eqn:dp}), we observe that for any $X\in\mathcal{A}$,
	\begin{equation}
	J_1(X) = -\lambda X+\max_{x\leq X, m} \{ -dm+xP_{m,x}(\mu^*)(1+\lambda)\}
	\end{equation}
	\noindent holds. If $m_1 = 0$, it implies that $-dm+xP_{m,x}(\mu^*)(1+\lambda)<0$ for all $m>0$, and $J_1(X) = J_0(X) = -\lambda X$. Identical situation arises in the next steps and the proof follows by induction.
\end{proof}

\begin{lemma}
\label{lemma:idle}
	If $\mu^* < \frac{d}{1+\lambda}$, then $m_t = 0$ for all $t$ and $X$.
\end{lemma}

\begin{proof}
	If $X=0$, then the result trivially holds. If $X>0$, Markov inequality provides the following result:
	\begin{equation}
	P_{m, x} (\mu^*)\leq \frac{m\mu^*}{x}
	\end{equation}
	
	\noindent for any $0<x\leq X$. Plugging this into the Bellman-Ford recursion at stage $s=1$,
	\begin{equation}
	J_0(X) \leq J_1(X) \leq J_0(X)+\max_{x\leq X, m} \{-dm+x\frac{m\mu^*}{x}(1+\lambda)\}
	\end{equation}
	
	\noindent which implies that $m_1 = 0$ if $\mu^*<\frac{d}{1+\lambda}$. By Lemma \ref{lemma:cont}, $m_t = 0$ for all $t$.
\end{proof}

\begin{lemma}
\label{lemma:nonidle}
If $\mu^* \geq \frac{2d}{1+\lambda}$, then $m_1 > 0$.
\end{lemma}

\begin{proof}
Setting $x = m\mu^*$ and using the fact that $x =m \mu^*$ is the median of Binomial distribution provides the result.
\end{proof}

\begin{lemma}
\label{lemma:conn}
Let $\Lambda = \{\mu^*\in[0,1]: \sum_{t=1}^T m_t = 0\}$. Then, $\Lambda$ is a connected set.
\end{lemma}

\begin{proof}
Take a non-zero $\mu_0\in\Lambda$. Then, by definition, for any $m > 0$, we know that $-dm+xP_{m,x}(\mu_0)(1+\lambda) < 0$. For any $\mu < \mu_0$, ${P}_{m,x}(\mu) < P_{m, x}(\mu_0)$ and therefore $-dm+xP_{m,x}(\mu)(1+\lambda) < 0$, which implies that no channel is used and thus $\mu\in\Lambda$.
\end{proof}

Lemma \ref{lemma:conn} states that the idle region is a continuum in the unit interval. With this, Lemma \ref{lemma:idle} and  Lemma \ref{lemma:nonidle} together establish subsets of idle and non-idle regions, respectively. Continuity of $P_{m, x}(\mu)$ provides the result.
\end{proof}

\subsection{Proof of Lemma \ref{lemma:01}}
\label{prf:01}
\begin{enumerate}
\item If $\mu \geq d$,
	\begin{eqnarray*}
	Z_0(n)&=&\sum_{n=1}^N\mathbb{I}_{\{ \bar{\mu}_{Z(n-1)}(n) < d \}} \\
	&\leq& \sum_{n=1}^N\mathbb{I}_{\{ \min\{ \bar{\mu}_{s}(n) < d :~1 \leq s \leq TM_{max}\cdot n\}\}} \\
	&\leq& \sum_{n=1}^N\sum_{s=1}^{TM_{max}\cdot n}\mathbb{I}_{\{ \bar{\mu}_{s}(n) < d \}} \\
	&\leq& \sum_{n=1}^N\sum_{s=1}^{TM_{max}\cdot n}\mathbb{I}_{\{ \bar{\mu}_{s}(n) < \mu^* \}}.
\end{eqnarray*}

by following a similar path as \cite{Bubeck2012}. Taking the expectation and using Chernoff-Hoeffding Bound, the following is obtained if $\beta \geq 3$:
\begin{eqnarray*}
	\mathbb{E}[Z_0(n)] &\leq& TM_{max}\sum_{n=1}^N n^{1-\beta} \\
	&\leq& TM_{max}\sum_{n=1}^\infty n^{1-\beta} \leq TM_{max}\frac{\pi^2}{6}.
\end{eqnarray*}

\item The following claim is necessary for proving this part.
\begin{claim}
If $\sum_t m_t > 0$, then at least one of the following must hold:
\begin{enumerate}
	\item $\hat{\mu}_{Z(n-1)} \geq \mu^*+c_{n, Z(n-1)}$
	\item $Z(n-1) < \frac{2\beta\log N}{(d-\mu^*)^2}$.
\end{enumerate}
\end{claim}
\begin{proof}[Proof of Claim 1]
Suppose neither holds. Then,
\begin{eqnarray*}
	\hat{\mu}_{Z(n-1)}+c_{n, Z(n-1)} &<& \mu^*+2c_{n,Z(n-1)} \\
	&\leq& \mu^*+2c_{n, Z(n-1)} \\
	&\leq& \mu^*+(d-\mu^*) = d.
\end{eqnarray*}
Thus, $\pi^*_{A(n)}(n) = (\mathbf{0},\mathbf{0})$.
\end{proof}

Let $l = \ceil*{\frac{2\beta\log N}{(d-\mu^*)^2}}$. Then, by using a similar methodology as \cite{Bubeck2012},
\begin{align*}
N-Z_0(N) &\leq \sum_{n=1}^N \mathbb{I}_{\{\bar{\mu}_{Z(n-1)}(n) > d, Z(n-1) \geq l\}} + l &&\\ 
&= \sum_{n=1}^N\mathbb{I}_{\{ \hat{\mu}_{Z(n-1)} \geq \mu+c_{n,Z(n-1)} \}} +l &&\\ 
&\leq \sum_{n=1}^N \sum_{s=1}^{TM_{max}n} \mathbb{I}_{\{ \hat{\mu}_{Z(n-1)} \geq \mu+c_{n,Z(n-1)} \}}+l,&&
\end{align*}

\noindent where the first line holds with equality iff $Z(n-1)\geq l$ and the second line follows from Claim 1. Taking the expectation and exploiting Chernoff-Hoeffding Bound, the result is obtained.
\end{enumerate}

\subsection{Proof of Lemma \ref{lemma:012}}

\begin{enumerate}
		\item Proof of this part is similar to the proof of the first part of Lemma \ref{lemma:01}.

		\item The following claims play an essential role in the proof.
		
		\begin{claim}
			\label{claim:I012}
			If $m_t=\tilde{m}_t,~\forall t$, then at least one of the following must hold:
			\begin{enumerate}
				\item $\hat{\mu}_{Z(n-1)} > \mu+c_{n, Z(n-1)}$,
				\item $c_{n, Z(n-1)} > \frac{\zeta_a^u-\mu^*}{2}$.
			\end{enumerate}
		\end{claim}

		\begin{claim}
			\label{claim:ET0}		
			For any $N\geq 1$, UCB-P with parameter $\beta \geq 4$ provides the following:
			\begin{equation*}
				\mathbb{E}[Z_0^2(N)] \leq TM_{max}\frac{\pi^2}{2}.
			\end{equation*}
		\end{claim}
		
		\begin{claim}
			\label{claim:P012}
			For any $\epsilon > 0$, UCB-P with $\beta \geq 4$ implies the following:
			\begin{equation*}
				\sum_{t=0}^\infty \bbP\Big(\frac{\log(n+1)}{Z(n)} > \epsilon\Big) \leq \Psi(\epsilon) < \infty
			\end{equation*}
		\end{claim}
		 The proofs for Claim \ref{claim:I012}, Claim \ref{claim:ET0} and Claim \ref{claim:P012} are given at the end of this subsection.

	Let $\tilde{Z}(N) = \sum_{n=1}^N\mathbb{I}_{\{\mathbf{m = \mathbf{\tilde{m}}}\}}$. Using Claim \ref{claim:I012}, $\tilde{Z}(N)$ can be upper bounded as follows:
	\begin{align*}
		\tilde{Z}(N) &\leq \sum_{n=1}^N\mathbb{I}_{\{c_{n, Z(n-1)} > \frac{\zeta-\mu^*}{2}\text{ or } \mh_{Z(n-1)}>\mu^*+c_{n,Z(n-1)}\}} & \\
		&\leq \sum_{n=1}^N\mathbb{I}_{\{c_{n, Z(n-1)} > \frac{\zeta-\mu^*}{2}\}} & \\  &\qquad+\sum_{n=1}^n\mathbb{I}_{\{\mh_{Z(n-1)}>\mu+c_{n,Z(n-1)} \}}\\
		&= \sum_{n=1}^N\mathbb{I}_{\{\frac{\log n}{Z(n-1)} > \frac{(\zeta-\mu^*)^2}{2\beta}\}}&& \\&\qquad+\sum_{n=1}^N\mathbb{I}_{\{\mh_{Z(n-1)}>\mu+c_{n,Z(n-1)} \}}.
	\end{align*}
	Thus, $\mathbb{E}[\tilde{Z}(N)]$ is upper bounded as follows:
	\begin{equation}
		\label{eqn:et2}
		\mathbb{E}[\tilde{Z}(N)] \leq \sum_{n=1}^N \bbP\Big(\frac{\log(n)}{Z(n-1)} > \frac{(\zeta-\mu^*)^2}{2\beta}\Big) + \frac{\pi^2}{3}.
	\end{equation}
	
	The first term on the right-hand side of (\ref{eqn:et2}) is upper bounded by Claim \ref{claim:P012} with $\epsilon = \frac{(\zeta-\mu^*)^2}{2\beta}$. Thus the proof follows.
	
		\end{enumerate}
	
	\begin{proof}[Proof of Claim \ref{claim:I012}]
		Suppose neither holds. Then,
		\begin{eqnarray*}
			\bar{\mu}_{Z(n-1)}+c_{n, Z(n-1)} &\leq& \mu+2\cdot c_{n, Z(n-1)} \\	
			&\leq& \mu^*+2\cdot\frac{\zeta-\mu^*}{2} = \zeta.
		\end{eqnarray*}
		Thus, $\textbf{m} \neq \mathbf{\tilde{m}}$.
	\end{proof}
	
	\begin{proof}[Proof of Claim \ref{claim:ET0}] The decomposition of $Z_0^2(N)$ into diagonal and off-diagonal elements and union bound provide the following upper bound:
		\begin{align*}
			Z_0^2(N) &\leq  \sum_{n=1}^N \mathbb{I}_{\{\bar{\mu}_{Z(n-1)}(n) \leq d\}}+2\sum_{n=1}^{N-1}\sum_{s=n+1}^N \mathbb{I}_{\{\bar{\mu}_{Z(s-1)}(s) \leq d\}}& \\
			 &\leq \sum_{n=1}^N \sum_{r=1}^{TM_{max}n} \mathbb{I}_{\{\bar{\mu}_{r}(n) \leq d\}}+2\sum_{n=1}^{N-1}\sum_{s=n+1}^N \mathbb{I}_{\{\bar{\mu}_{r}(s) \leq d\}}.&
		\end{align*}
		Taking the expectation, and applying Chernoff-Hoeffding Bound,
		\begin{align*}
			\mathbb{E}Z_0^2(N) &\leq TM_{max}\sum_{n=1}^N n^{1-\beta}+2TM_{max}\sum_{n=1}^{N-1}\sum_{s=n+1}^N s^{1-\beta}& \\
			&\leq TM_{max}\sum_{n=1}^N n^{1-\beta}+2TM_{max}\cdot\sum_{n=1}^{N}n^{2-\beta}& \\
			&\leq TM_{max}(\frac{\pi^2}{6}+\frac{\pi^2}{3}) = TM_{max}\frac{\pi^2}{2}.&
		\end{align*}
	\end{proof}		
	
	\begin{proof}[Proof of Claim \ref{claim:P012}]
	Fix $\epsilon > 0$. Note that 
	\begin{eqnarray*}
		Z(N)&=&N-Z_0(N)+Z_2(N) \\ &\geq& N-Z_0(N). 
	\end{eqnarray*}		
	Therefore,
	\begin{equation}
		\label{eqn:ubP012}
		\bbP\Big(\frac{\log(n+1)}{Z(n)} > \epsilon\Big) \leq \bbP\Big(Z_0(n) > n-\frac{\log(n+1)}{\epsilon}\Big).
	\end{equation}
	
	If $t-\frac{\log(n+1)}{\epsilon} > 0$, Markov Inequality applied to the RHS of (\ref{eqn:ubP012}) implies the following:
	\begin{equation}
		\bbP\Big(Z_0(n) > n-\frac{\log(n+1)}{\epsilon}\Big) \leq \frac{\mathbb{E}Z_0^2(n)}{(n-\frac{\log(n+1)}{\epsilon})^2}.
	\end{equation}
	Let $n_\epsilon = \inf\{n:n-\frac{\log(n+1)}{\epsilon}>0\}$. Then, by Claim \ref{claim:ET0},
	
	\begin{align*}
		\sum_{n=0}^\infty \bbP\Big(\frac{\log(n+1)}{Z(n)} > \epsilon\Big) &\leq  \Psi(\epsilon)< \infty.
	\end{align*}

\end{proof}

%

%
%
%
%
%
%
%
%

%
%


\begin{IEEEbiographynophoto}{Semih Cayci} (S'12) received the B.S. degree from Bogazici University in 2010 and M.S. degree from Bilkent University in 2013, both in Electrical and Electronics Engineering. Currently, he is working towards the Ph.D. degree in the Department of Electrical and Computer Engineering of the Ohio State University.  His research interests include machine learning, probability theory and stochastic control.
\end{IEEEbiographynophoto}

\begin{IEEEbiographynophoto}{Atilla Eryilmaz}
(S'00 / M'06 / SM'17 ) received his
M.S. and Ph.D. degrees in Electrical and Computer
Engineering from the University of Illinois at
Urbana-Champaign in 2001 and 2005, respectively.
Between 2005 and 2007, he worked as a Postdoctoral
Associate at the Laboratory for Information and
Decision Systems at the Massachusetts Institute of
Technology. He is currently an Associate Professor
of Electrical and Computer Engineering at The Ohio
State University. Dr. Eryilmaz's research interests include design and analysis for communication
networks, optimal control of stochastic networks, optimization theory,
distributed algorithms, pricing in networked systems, and information theory. He received the NSF-CAREER Award in 2010 and two Lumley Research Awards for Research Excellence in 2010 and 2015. He is a co-author of the 2012 IEEE WiOpt Conference Best Student Paper, the 2016 IEEE Infocom Best Paper, and the 2017 IEEE WiOpt Conference Best Paper Awards. He has served as TPC co-chair of IEEE WiOpt in 2014 and of ACM Mobihoc in 2017, and is an Associate Editor of IEEE/ACM Transactions on Networking and IEEE Transactions on Networks Science and Engineering.
\end{IEEEbiographynophoto}

\end{document}